\documentclass[english,pra,superscriptaddress,reprint,twocolumn]{revtex4-1}
\usepackage{standalone}
\usepackage{nccfoots}
\usepackage{comment}
\usepackage{graphics}
\graphicspath{{Figures/}}
\usepackage{braket}
\usepackage{color}
\usepackage{nicefrac}

\usepackage[colorlinks=true,
   linkcolor=blue,
   citecolor=blue,
   urlcolor =blue]{hyperref}

\newcommand{\citeSM}{\cite[{\tiny SM}\kern-0.3em][]{SM}}

\newcommand{\be}{\begin{equation}}
\newcommand{\ee}{\end{equation}}

\newcommand{\tr}{\mathrm{Tr}}

\expandafter\let\csname equation*\endcsname\relax

\expandafter\let\csname endequation*\endcsname\relax

\usepackage{amsmath,amsfonts,amssymb,amsthm, bbm,braket}
\usepackage{tikz}
\usepackage{xcolor}

\newtheorem{proposition}{Proposition}
\newtheorem{lemma}{Lemma}
\newtheorem{fact}{Fact}

\definecolor{Red}{HTML}{E53E30}  
\definecolor{Green}{HTML}{00AD69}  
\definecolor{Blue}{HTML}{2171b5}
\definecolor{Purple}{HTML}{652F6C}  

\begin{document}
\title{Mixed-state entanglement from local randomized measurements}

\author{Andreas Elben}
\thanks{These authors contributed equally.}
\affiliation{Center for Quantum Physics, University of Innsbruck, Innsbruck A-6020, Austria}
\affiliation{Institute for Quantum Optics and Quantum Information of the Austrian Academy of Sciences,  Innsbruck A-6020, Austria}

\author{Richard Kueng}
\thanks{These authors contributed equally.}
\affiliation{Institute for Integrated Circuits, Johannes Kepler University Linz, Altenbergerstrasse 69, 4040 Linz, Austria}

\author{Hsin-Yuan (Robert) Huang}
\affiliation{Institute for Quantum Information and Matter, Caltech, Pasadena, CA, USA}
\affiliation{Department of Computing and Mathematical Sciences, Caltech, Pasadena, CA, USA}

\author{Rick van Bijnen}
\affiliation{Center for Quantum Physics, University of Innsbruck, Innsbruck A-6020, Austria}
\affiliation{Institute for Quantum Optics and Quantum Information of the Austrian Academy of Sciences,  Innsbruck A-6020, Austria}

\author{Christian Kokail}
\affiliation{Center for Quantum Physics, University of Innsbruck, Innsbruck A-6020, Austria}
\affiliation{Institute for Quantum Optics and Quantum Information of the Austrian Academy of Sciences,  Innsbruck A-6020, Austria}

\author{Marcello Dalmonte}
\affiliation{The Abdus Salam International Center for Theoretical Physics, Strada Costiera 11, 34151 Trieste, Italy}
\affiliation{SISSA, via Bonomea 265, 34136 Trieste, Italy}

\author{Pasquale Calabrese}
\affiliation{The Abdus Salam International Center for Theoretical Physics, Strada Costiera 11, 34151 Trieste, Italy}
\affiliation{SISSA, via Bonomea 265, 34136 Trieste, Italy}
\affiliation{INFN, via Bonomea 265, 34136 Trieste, Italy}

\author{Barbara Kraus}
\affiliation{Institute for Theoretical Physics, University of Innsbruck, A–6020 Innsbruck, Austria}

\author{John Preskill}
\affiliation{Institute for Quantum Information and Matter, Caltech, Pasadena, CA, USA}
\affiliation{Department of Computing and Mathematical Sciences, Caltech, Pasadena, CA, USA}
\affiliation{Walter Burke Institute for Theoretical Physics, Caltech, Pasadena, CA, USA}
\affiliation{AWS  Center  for  Quantum  Computing,  Pasadena,  CA, USA}

\author{Peter Zoller}
\affiliation{Center for Quantum Physics, University of Innsbruck, Innsbruck A-6020, Austria}
\affiliation{Institute for Quantum Optics and Quantum Information of the Austrian Academy of Sciences,  Innsbruck A-6020, Austria}

\author{Beno\^it Vermersch}
\affiliation{Center for Quantum Physics, University of Innsbruck, Innsbruck A-6020, Austria}	
\affiliation{Institute for Quantum Optics and Quantum Information of the Austrian Academy of Sciences,  Innsbruck A-6020, Austria}
\affiliation{Univ.  Grenoble Alpes, CNRS, LPMMC, 38000 Grenoble, France}

\begin{abstract}
	We propose a method for detecting bipartite entanglement in a many-body mixed state based on estimating moments of the partially transposed density matrix. The estimates are obtained by performing local random measurements on the state, followed by post-processing using the classical shadows framework. Our method can be applied to any quantum system with single-qubit control. We provide a detailed analysis of the required number of experimental runs, and demonstrate the protocol using existing experimental data [Brydges et al, Science \textbf{364}, 260 (2019)].
\end{abstract}
\maketitle

Engineered quantum many-body systems exist in today's laboratories as Noisy Intermediate Scale Quantum Devices (NISQ) \cite{Preskill2018}. This provides us with novel opportunities to study and quantify  entanglement -- a fundamental
concept in both quantum information theory~\cite{Horodecki2009} and many-body quantum physics \cite{Eisert2010,Calabrese2016}.
For pure (or nearly-pure) states, entanglement has been detected by measuring the second R\'{e}nyi entropy ~\cite{Horodecki1996,Horodecki2003,Islam2015,Kaufman2016, Linke2018,Brydges2019}.
This has been achieved via, for instance, many-body quantum interference~\cite{Alves2004,Daley2012, Islam2015,Kaufman2016, Linke2018} (see also \cite{Cardy2011,Abanin2012}) and randomized measurements~\cite{VanEnk2012,Elben2018,Elben2018a,Brydges2019,Huang2020}.
However, many states of interest are actually highly mixed -- either because of decoherence, or because they describe interesting subregions of a larger, globally entangled, system. Developing protocols which detect and quantify mixed-state entanglement on intermediate scale quantum devices is thus an outstanding challenge.

Below we propose and experimentally demonstrate conditions for mixed-state entanglement and measurement protocols based on  the positive partial transpose (PPT) condition~\cite{Peres1996,Horodecki1996,Horodecki2009}.
Consider two partitions $A$ and $B$ described by a (reduced) density matrix $\rho_{AB}$. The well-known \emph{PPT condition}
checks if the partially transposed (PT)  density matrix $\rho_{AB}^{T_A}$ \footnote{The partial transpose (PT) operation -- acting on subsystem $A$ -- is defined as $(\ket{k_A,k_B}\bra{l_A,l_B})^{T_A}=\ket{l_A,k_B}\bra{k_A,l_B}$, where $\{\ket{k_A,k_B}\}$ is a product basis of the joint system $AB$.} is positive semidefinite, i.e.\ all eigenvalues are non-negative.
If the PPT condition is violated -- i.e.~$\rho_{AB}^{T_A}$ does have negative eigenvalues -- $A$ and $B$ must be entangled.
It is possible to turn the PPT condition into a quantitative entanglement measure. The \emph{negativity} $\mathcal{N}(\rho_{AB})=\sum_{\lambda<0}|\lambda|$, with $\lambda$ the spectrum of $\rho_{AB}^{T_A}$, is positive if and only if the underlying state $\rho_{AB}$ violates the PPT condition~\cite{Vidal2001}.
While applicable to mixed states, computing the negativity requires accurately estimating the full spectrum of $\rho_{AB}^{T_A}$.
We bypass this challenge by considering  moments of the partially transposed density matrix (PT-moments)  instead:
\begin{align}
	p_n=\mathrm{Tr}[(\rho_{AB}^{T_A})^n] \quad \text{for $n=1,2,3,\ldots$}.
	\label{eq:negativities}
\end{align}
These have been first  studied in quantum field theory to quantify correlations in many-body systems~\cite{Calabrese2012}. Clearly, $p_1 
=\mathrm{tr}(\rho_{AB})=1$, while $p_2$ 
is equal to the purity $\mathrm{tr}[ \rho_{AB}^2]$ (see Table~1 in the Supplemental Material~\cite{SM} 
 (SM) for a visual derivation). Hence, $p_3$ is the lowest PT-moment that captures meaningful information about the partial transpose  (see also Ref.~\cite{Gray2017}).

In this letter, we first show that the first three PT-moments can be used to define a simple yet powerful test for bipartite entanglement:
\begin{equation}
	\rho_{AB} \in \mathrm{PPT} \implies p_3\geq p_2^{2}.
	\label{eq:p3PPT}
\end{equation} 
The \emph{$p_3$-PPT condition} is the contrapositive of this assertion: if $p_3 < p_2^2$, then $\rho_{AB}$ violates the PPT condition [see Fig.~\ref{fig:setup}a)] and must therefore be entangled (see SM  \cite{SM} for the proof).
Similar to the PPT condition, the $p_3$-PPT condition applies to mixed states and is completely independent of the state in question. 
This is a key distinction from entanglement witnesses \cite{Terhal2000,Guhne2006},  which can be more powerful, but which usually require detailed prior information about the state.
While in general weaker than the full PPT condition, the $p_3$-PPT condition relies on comparing two comparatively simple functionals and outperforms other state-independent entanglement detection protocols, like comparing purities of various nested subsystems~\cite{Horodecki1996,Islam2015,Kaufman2016,Linke2018,Brydges2019,SM}.  
 As shown in the SM~\cite{SM}, the $p_3$-PPT condition becomes equivalent to the PPT condition for Werner states (in this case, it is a necessary and sufficient condition for bipartite entanglement \cite{Watrous2018}).

The second main contribution of this letter is a measurement protocol to determine PT-moments in NISQ devices.
\begin{figure}
	\includegraphics[width=1.\linewidth]{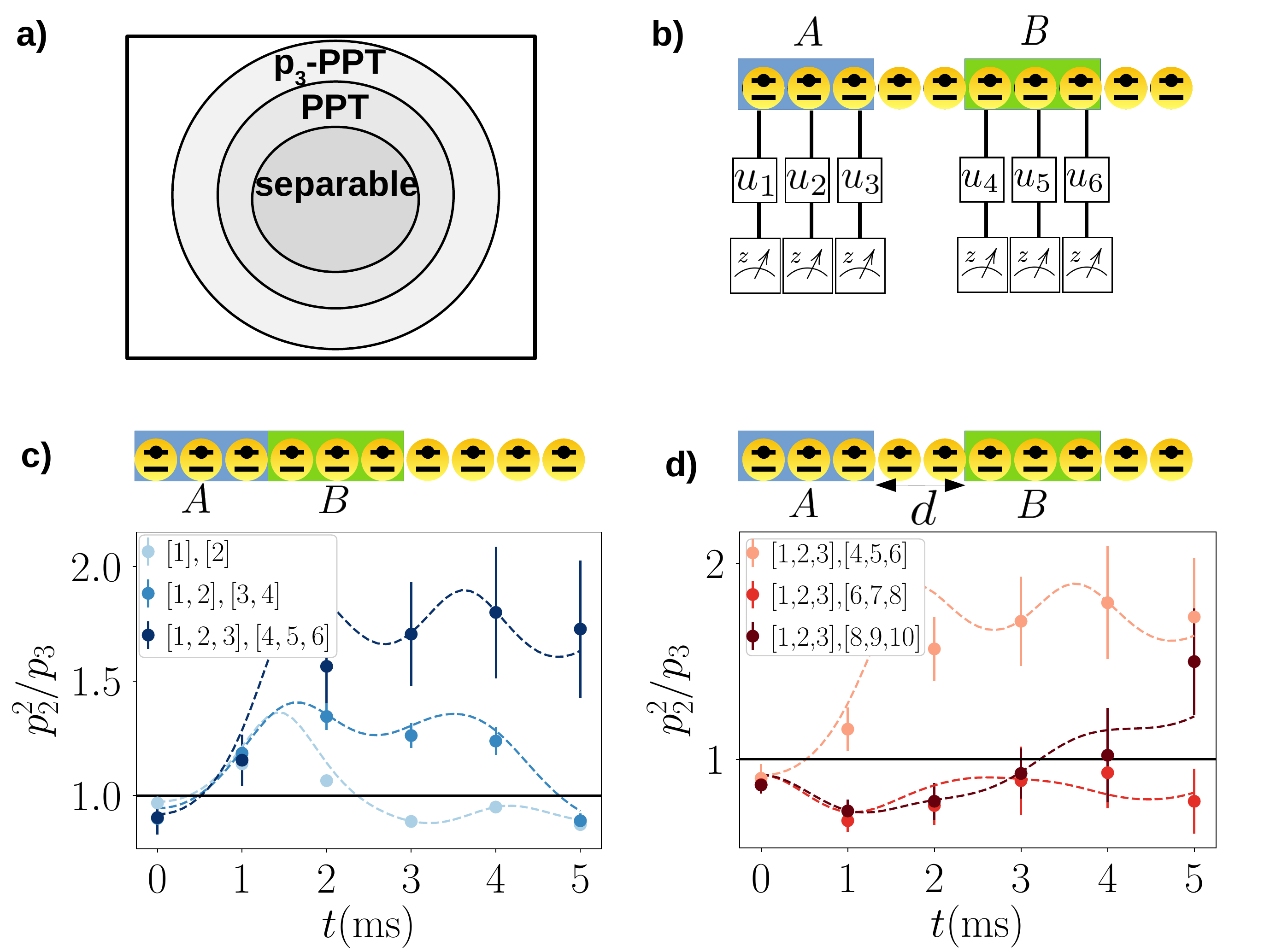}
	\caption{{\it Protocol and illustrations}.
		a) The $p_3$-PPT condition  can  be used to demonstrate mixed-state bipartite entanglement with	PT-moments.
		Separable states  are PPT states and also fulfill the $p_3$-PPT condition. Thus, quantum states which violate the $p_3$-PPT condition must be bipartite entangled [see also Eq.~\eqref{eq:p3PPT}].
		b) In our protocol, PT-moments are measured by applying local random unitaries followed by computational basis measurements. 
		c-d) Violation of the $p_3$-PPT condition, i.e.~$p_2^2>p_3$, is experimentally observed for connected c) and disconnected (separated by $d=0,2,4$ spins) d) partitions $A$ and $B$  at various times $t$  after a quantum quench~\cite{Brydges2019}. Dots: experimental results. Error bars: Jackknife estimates of statistical errors. Lines: numerical simulations including the decoherence model presented in Ref.~\cite{Brydges2019}.
	}
	\label{fig:setup}
\end{figure}
Crucially, we employ randomized measurements implemented with \emph{local} (single-qubit) random unitaries, see Fig.~\ref{fig:setup}b) which are readily available in NISQ devices and have  been already successfully applied to measure entanglement entropies, many-body state-fidelities, and out-of-time ordered correlators~\cite{Brydges2019,Mi2020,Elben2020,Joshi2020}.
In contrast to previous proposals for measuring PT-moments, our protocol does not rely on  many-body interference between identical state copies~\cite{Horodecki2003,Gray2017,Cornfeld2019}, or on using global entangling random unitaries~\cite{Zhou2020} built from interacting Hamiltonians~\cite{Dankert2009,Nakata2017,Elben2018,Vermersch2018}. Instead, it only requires single-qubit control, and allows for the estimation of many distinct PT-moments from the same data.
In particular, arbitrary orders $n \geq 2$ and arbitrary (connected, as well as disconnected) partitions $A$, $B$ can be measured.

While the experimental setup for our measurement protocol
is reminiscent of quantum state tomography~\cite{Gross2010,Cramer2010,Torlai2018,Guta2020}, there are fundamental differences regarding the required number of measurements (as independent state copies), and the way the measured data is processed. Without strong assumptions on the state~\cite{Cramer2010,Torlai2018},
performing tomography to infer an $\epsilon$-approximation of an unknown density matrix $\rho_{AB}$ (e.g.\ in order to subsequently compute $\epsilon$-approximations of $p_n$) requires (at least) order $2^{|AB|}\mathrm{rank}(\rho_{AB})/\epsilon^2$ measurements~\cite{Haah2017,Wright2016}. In the high accuracy regime ($\epsilon \ll 1$), our direct estimation protocol instead only requires  order $2^{|AB|}/\epsilon^2$ measurements.
For highly mixed states -- the central topic of this work -- this discrepancy heralds a significant reduction in measurement resources.
Furthermore, we predict PT-moments through a 'direct' and (multi-) linear postprocessing of the measurement data represented as 'classical shadows'~\cite{Huang2020}. Thus, data processing
is cheap -- both in memory and runtime -- and can be massively parallelized.  Similar to previous measurement \cite{VanEnk2012,Elben2018,Brydges2019,Vermersch2019,Joshi2020,Elben2020,Elben2020b,Huang2020,Cian2020} and entanglement detection \cite{Tran_2015,Tran_2016,Ketterer2019,Zhang_2020, Knips_2020}  protocols based on randomized measurements, this is another distinct advantage over tomography which typically requires expensive data-processing algorithms~\cite{Gross2010} or training a neural network~\cite{Torlai2018}.

Finally, we demonstrate our measurement protocol and the $p_3$-PPT condition experimentally in the context of the quantum simulation of many-body systems.   Here, PT-moments have been shown to reveal universal properties of quantum phases of matter~\cite{Calabrese2012,Calabrese2013,Ruggiero2016,Javanmard2018,Turkeshi2019} and their transitions~\cite{Calabrese2012,Calabrese2013,Chung2014,Wu2019}. Out of equilibrium, PT-moments allow to understand  the dynamical process of thermalization~\cite{Coser2014,Alba2019,Kudler-Flam2020,Alba2020}, and the fate of (many-body) localization in presence of decoherence~\cite{Wybo2020}.
In this work, we analyze the data of Ref.~\cite{Brydges2019} corresponding to the out-of-equilibrium dynamics in a spin model with long-range interactions, which was implemented in a $10$-qubit trapped ion quantum simulator.  In particular, we certify the presence of mixed-state entanglement via the $p_3$-PPT condition [see Fig.~\ref{fig:setup}(c-d), and for details below]. Furthermore, we  monitor the time-evolution of $p_3$ and observe dynamical signatures of entanglement spreading and thermalization~\cite{Coser2014,Alba2019}.

{\it Protocol--}
The experimental ingredients to measure PT-moments build on resources similar to the ones presented in Ref.~\cite{Elben2018} and realized in Ref.~\cite{Brydges2019} to measure R\'enyi entropies. The key new element is the  post-processing of the experimental data \cite{Huang2020}. As shown in Fig.~\ref{fig:setup},
the quantum state of interest is realized in a system of $N$ qubits. In the partitions $A$ and $B$, consisting of $|A|$ and $|B|$ spins, respectively, a randomized measurement is performed by applying  random local unitaries $u=u_1\otimes \dots \otimes u_{|AB|}$, with $u_{i}$ independent single qubit rotation sampled from a unitary $3$-design \cite{Gross2007,Dankert2009}, and a subsequent projective measurement in the computational basis with outcome $\mathbf{k}=(k_1,\dots, k_{|AB|})$. This is subsequently repeated with $M$ different random unitaries such that a data set of $M$ bitstrings $\mathbf{k}^{(r)}$ with $r = 1, \dots, M$ is collected.

From this data set, the PT-moments  $p_n$ can be estimated \emph{without} having to reconstruct the density matrix $\rho_{AB}$, 
and with a  significantly
smaller number of experimental runs $M$ than required for full quantum state tomography. To obtain such estimates, we rely on two observations. First, each outcome $\mathbf{k}^{(r)}$ can be used to define an unbiased estimator 
\begin{equation}
	\hat{\rho}_{AB}^{(r)}=\bigotimes_{i \in AB} \left[ 3 (u_{i}^{(r)})^\dag \ket{k_{i}^{(r)}}\bra{k_{i}^{(r)}}u_{i}^{(r)}-\mathbb{I}_2\right]\label{eq:rho_s}
\end{equation}
of the density matrix $\rho_{AB}$, i.e.\ $\mathbb{E}[\hat \rho_{AB}^{(r)}] = \rho_{AB}$ with the expectation value taken over the unitary ensemble and projective measurements~\cite{Ohliger_2013,Elben2018a,Paini2019,Huang2020}.  Second, the PT-moments $p_n$ can be viewed as an  expectation value  of a $n$-copy observable $ \overrightarrow{\Pi}_{A}\, \overleftarrow{\Pi}_{B}$ evaluated on $n$-copies of the original density matrix $\rho_{AB}$,
\begin{equation}
	p_n=\mathrm{Tr}\left[ \overrightarrow{\Pi}_{A}\, \overleftarrow{\Pi}_{B} \, \rho_{AB}^{\otimes n}\right].
	\label{eq:pnmulticopy}
\end{equation}
Here, $\overrightarrow{\Pi}_{A}$ and $\overleftarrow{\Pi}_{B}$ are  $n$-copy  cyclic permutation operators  $\overrightarrow{\Pi}_A\ket{\mathbf{k}^{[1]}_A,\mathbf{k}^{[2]}_A,\dots,\mathbf{k}^{[n]}_A}=\ket{\mathbf{k}^{[n]}_A,\mathbf{k}^{[1]}_A,\dots,\mathbf{k}^{[n-1]}_A}$, $\overleftarrow{\Pi}_B\ket{\mathbf{k}^{[1]}_B,\mathbf{k}^{[2]}_B,\dots,\mathbf{k}^{[n]}_B}=\ket{\mathbf{k}^{[2]}_B,\dots,\mathbf{k}^{[n]}_B,\mathbf{k}^{[1]}_B}$ that act on the partitions $A$ and $B$,  respectively.

Estimators of the PT-moments $p_n$ can now be derived from Eqs.~\eqref{eq:rho_s} and \eqref{eq:pnmulticopy} using U-statistics \cite{Hoeffding1992}. Replacing $\rho^{\otimes n}$ with  $\hat{\rho}^{(r_1)}\otimes \dots \otimes \hat{\rho}^{(r_n)}$ where $r_1\neq r_2 \neq \dots \neq r_n$, corresponding to independently sampled random unitaries $u^{(r_1)},\dots,u^{(r_n)}$, we define the U-statistic
\begin{align}
	\hat{p}_n=  \frac{1}{n!}\binom{M}{n}^{-1} \! \! \!
	\sum_{\substack{r_1\neq r_2\neq \dots \neq r_n}} \!\!\! \mathrm{Tr}\left[ \overrightarrow{\Pi}_{A}\, \overleftarrow{\Pi}_{B}  \hat \rho_{AB}^{(r_1)} \otimes \dots \otimes \hat \rho_{AB}^{(r_n)}\right] .\label{eq:p_n}
\end{align}
It follows from the defining properties of $U$-statistics that $\hat{p}_n$ is  an unbiased estimator of $p_n$, i.e.~$\mathbb{E}[\hat{p}_n] = p_n$ with the expectation value taken over the unitary ensemble and projective measurements \cite{Hoeffding1992}. Its variance governs the statistical errors arising from finite $M$. Furthermore, a quick inspection of  Eqs.~\eqref{eq:rho_s} and \eqref{eq:pnmulticopy} reveals that the summands in  Eq.~\eqref{eq:p_n} completely factorize into contractions of single qubit  matrices, $\mathrm{Tr}[ \overrightarrow{\Pi}_{A}\, \overleftarrow{\Pi}_{B}  \hat \rho_{AB}^{(r_1)} \otimes \dots \otimes \hat \rho_{AB}^{(r_n)}] = \Pi_{i \in A} \tr[\hat \rho_i^{(r_1),T} \cdots \hat \rho_i^{(r_n),T}  ]\Pi_{i \in B} \tr[\hat \rho_i^{(r_1)} \cdots\hat \rho_i^{(r_n)}  ]$,  with $\hat\rho_{AB}^{(r)}=\otimes_{i \in AB}\hat\rho^{(r)}_i$ as in Eq.~\eqref{eq:rho_s}. Thus, given $M$ observed bitstrings $\mathbf{k}_r$, one can determine $\hat p_n$  with  classical data processing scaling as $M^n |AB|$, without storing exponentially large  matrices on the classical post-processing device.

\begin{figure}
	\includegraphics[width=1.\linewidth]{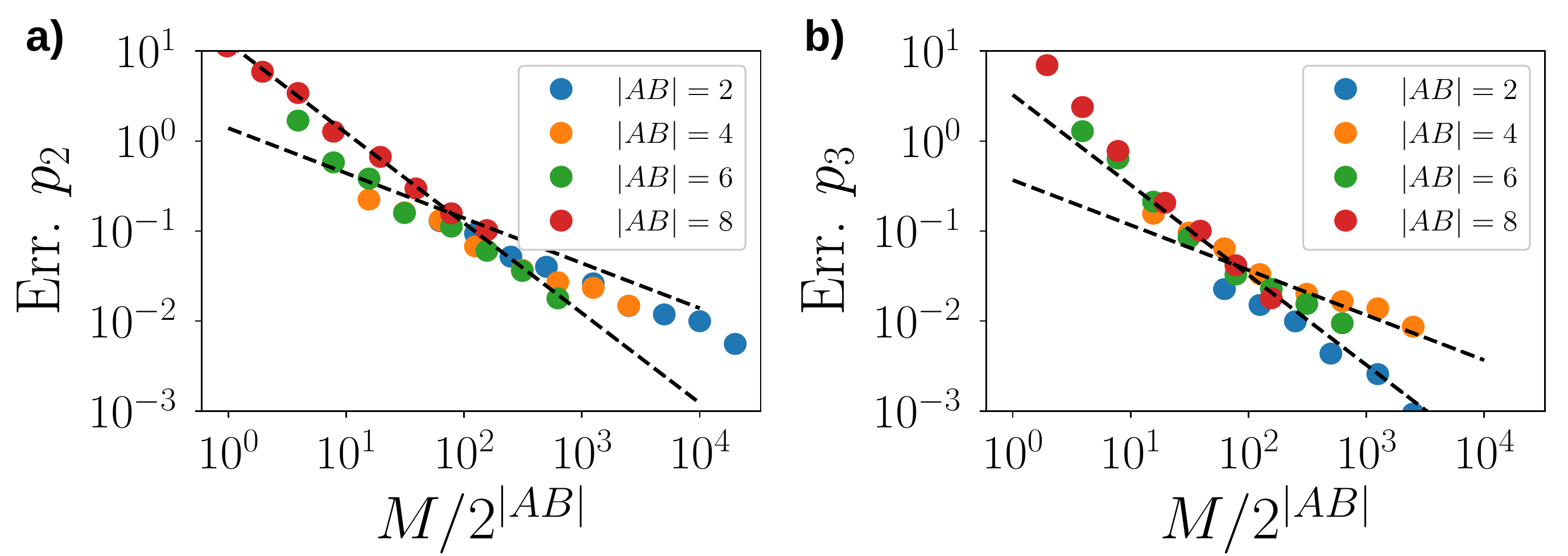}
	\caption{Statistical errors for the GHZ state. Dashed lines represent scalings of  $\propto 1/M$, and $\propto1/\sqrt{M}$. In both cases, the number of measurements to
		estimate $p_2$ a) and $p_3$ b) with accuracy $~0.1$ is of the order of  $100\times 2^{|AB|}$.}
	\label{fig:staterrors}
\end{figure}

{\it Statistical errors--}
As demonstrated in Fig.~\ref{fig:setup}(c,d), PT-moments can be inferred using a finite number of experimental runs $M$. 
Here, we investigate in detail the statistical errors arising from the finite value of  $M$.

Analytically, we bound statistical errors based on the variance of the multi-copy observable in question. For $p_2=\mathrm{Tr} [(\rho_{AB}^{T_A})^2]$, our analysis reveals that the error decay rate depends on 
number of measurements $M$. In the large $M$ regime, the error is proportional to $2^{|AB|}p_2/\sqrt{M}$. This error bound is multiplicative -- i.e.\ the size of the error is proportional to the size of the target $p_2$ -- and  $1/\sqrt{M}$ captures the expected decay rate for an estimation procedure that relies on empirical averaging. For small and intermediate values of $M$, the estimation error is instead bounded by $8 \times 2^{1.5|AB|}/M$. While this is worse in terms of constants, the error decays at a much faster rate proportional to $1/M$. 
Qualitatively similar results apply for estimating $p_3 = \mathrm{Tr}[ (\rho_{AB}^{T_A})^3]$, but there can be three decay regimes.
For large $M$, the estimation error is bounded by $2^{|AB|}p_2^2/\sqrt{M}$. This again captures the asymptotically optimal rate $1/\sqrt{M}$ associated with empirical averaging, but the constant is suppressed by $p_2^2$, not $p_3$ itself. For intermediate $M$, the error decay rate is proportional to $1/M$, while an even faster rate $\propto 1/M^{3/2}$ governs the error decay for small $M$. We refer to the SM for detailed statements and proofs. 

Now, we test these predictions numerically by simulating the experimental protocol for various values of $M$ in systems with $N=|AB|$ qubits where a pure GHZ state $\rho=\ket{\phi_{\text{GHZ}}}\!\bra{\phi_{\text{GHZ}}}$ is prepared. Here, $A$ corresponds to the first $N/2$ qubits, and $B$ is the complement. The results are shown in Fig.~\ref{fig:staterrors} and support our analytical error bounds. They  highlight in particular that the number of measurement repetitions necessary to achieve a desired accuracy of $\sim0.1$ scales as $2^{|AB|}$. This enables the estimation 
of PT-moments in state of the art platforms with high repetition rates.  These findings are discussed and confirmed for the ground state of the transverse Ising model in the  SM  \cite{SM}. 
\begin{figure}
	\includegraphics[width=0.99\columnwidth]{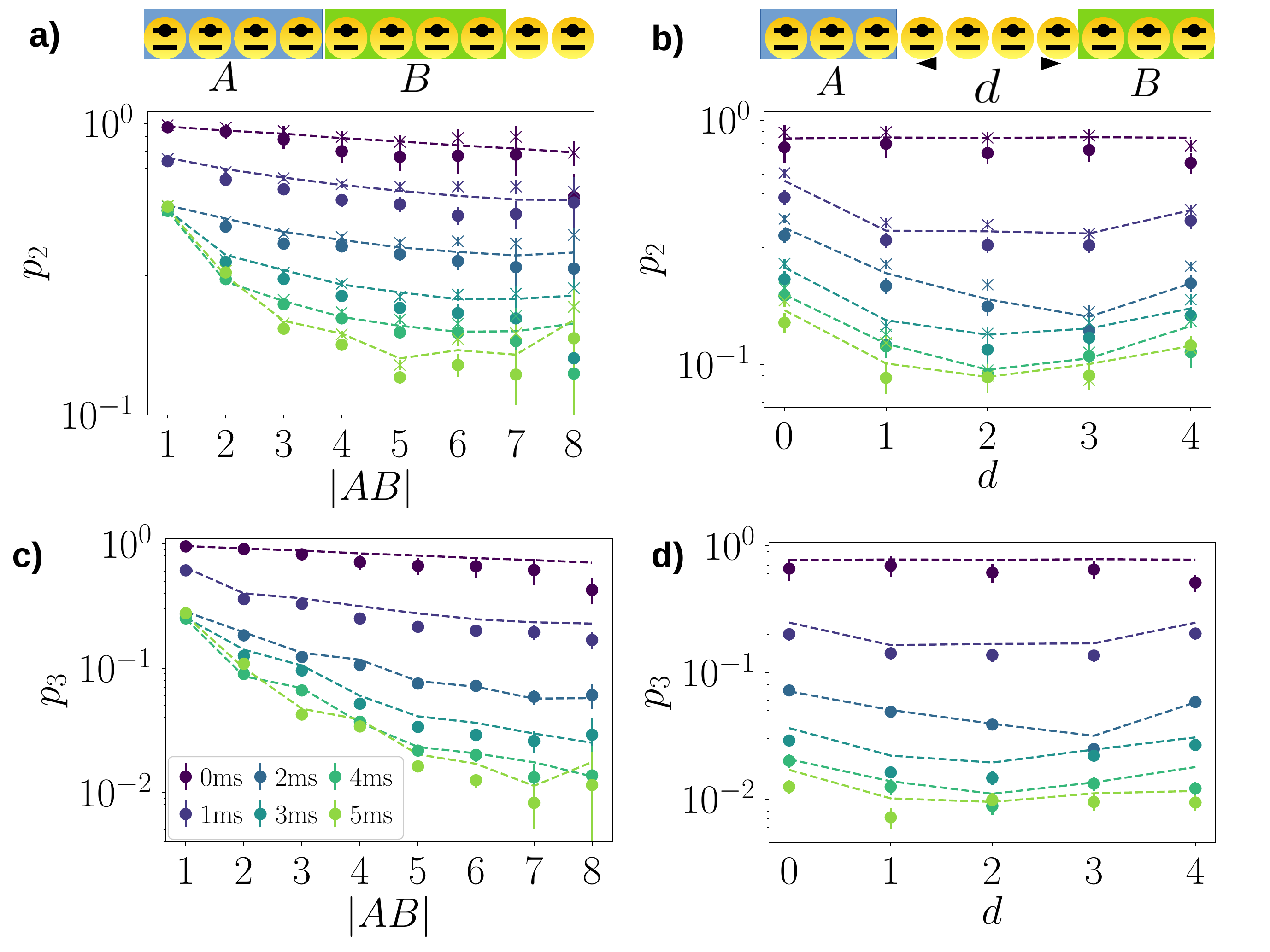}
	\caption{{\it Reconstruction of $p_2=\mathrm{Tr} [(\rho_{AB}^{T_A})^2]$ and $p_3=\mathrm{Tr} [(\rho_{AB}^{T_A})^3]$ from experimental data~\cite{Brydges2019}}. $A$ and $B$ are parts of a total system of  10 qubits. In 
		a) and c), we take $A= [ 1,\dots,\lfloor |AB|/2 \rfloor]$  and $B=[\lfloor |AB|/2 \rfloor +1, \dots, |AB|]$. In 
		b) and d), we take  $A=\{1,2,3\}$ and $B=\{4+d,5+d,6+d\}$ with $d=0,1,\dots 4$.
		Dots are obtained with the shadow estimator [Eq.~\eqref{eq:p_n} second and third order], crosses with the direct  estimator (second order) of Ref.~\cite{Brydges2019}. Different colors correspond to different times after the quantum quench with purple [0\,ms] corresponding to the initial product state. For each time, $M=500$ unitaries and $P=150$ measurements per unitary were used. Lines: theory simulation including  decoherence~\cite{Brydges2019}. The ratio $p_2^2/p_3$, detecting entanglement according to the $p_3$-PPT condition, is shown in Fig.~\ref{fig:setup} c) and d).}
	\label{fig:exp}
\end{figure}

{\it PT-moments in a trapped-ion quantum simulator--}

Below, we discuss the experimental demonstration of the measurement of PT-moments in a trapped ion quantum simulator. To this end, we evaluate data taken in the context of Ref.~\cite{Brydges2019}. Here, the R\'enyi entropy growth in quench dynamics was investigated. 
The system, consisting in total of $N=10$ qubits, was initialized in the N\'eel state $\ket{\uparrow \downarrow \uparrow \downarrow \dots}$, and time-evolved with 
\begin{equation}
	H_{\mathrm{XY}} = \hbar\sum_{i<j}J_{ij}(\sigma^{+}_{i}\sigma^{-}_{j}+\sigma^{-}_{i}\sigma^{+}_{j}) + \hbar B \sum_{i}\sigma^{z}_{i}
	\label{eq:XY-Hamiltonian}
\end{equation}
with  $\sigma_i^{z}$  the third spin-$1/2$ Pauli operator, $\sigma_i^{+}(\sigma_i^{-})$ the spin-raising (lowering) operators acting on spin $i$, and $J_{ij} \approx J_0/ \lvert{i-j}\rvert^{\alpha}$ the coupling matrix with an approximate power-law decay $\alpha\approx 1.24$ and $J_0=420 s^{-1}$.  After time evolution, randomized measurements were performed, using $M=500$ random unitaries  and $P=150$ projective measurements per random unitary. 

From this data, PT-moments can be inferred~\footnote{Theoretically, an ideal distribution of the measurement budget $M\cdot P$ would consist in setting $P=1$, i.e. sampling new random unitaries for each experimental run. Experimentally, it might be however beneficial to use $P\geq 1$. In this situation, we replace the estimators defined in Eq.~(3) with $\hat{\rho^{(r)}} = \sum_{s=1}^{P} \hat{\rho}^{(r,s)} $  where $\hat{\rho}^{(r,s)}=\bigotimes_i (3(u_i^{(r)})^\dag \ket{k_{i}^{(r,s)}}\bra{k_{i}^{(r,s)}}u_i^{(r)}-\mathbb{I}_2)$ and $k_{i}^{(r,s)}$ the outcome of the measurement $s$ obtained after the application of the unitary $r$. }, with results presented in Fig.~\ref{fig:exp}. For the purity $p_2$ a) b), we observe good agreement with theory for up to $N=8$ qubits partitions, in particular the raise of $p_2$ for partition sizes approaching the total system size which is expected for such nearly pure states. For $9$, $10$-qubit partitions, the data is not shown since the relative statistical error of the estimated data points approaches unity \footnote{The distribution of the measurement budget $M\cdot P$ in Ref.~\cite{Brydges2019} into unitaries $M$ and projective measurements per unitary $P$ has been optimized for the purity estimator $\tilde p_2$ presented in Ref.~\cite{Brydges2019} which differs from $\hat p_2$ defined in Eq.~\eqref{eq:p_n}. 
	Thus, for the present data set \cite{Brydges2019} with $M=500$ and $P=150$, the statistical uncertainty of the $\tilde p_2$ is smaller than for $\hat p_2$ which performs best for $P=\mathcal{O}(1)$ and a correspondingly larger number of unitaries $M$.}. We however note that the measured $\hat p_2$ is  slightly underestimated. This is due to imperfect realizations of the random unitaries, which tend to reduce the estimation of the overlap $\mathrm{Tr}(\rho_{r_1}\rho_{r_2})$. This effect is also present when measuring cross-platform fidelities~\cite{Elben2020}. For the third PT-moment $p_3$ c), d), we observe the same kind of agreement between theory value and experimental measurements. In particular, at large partition sizes, the protocol is able to measure with high precision  small values of $p_3$. These small values are indeed fundamental to detect entanglement:  a PPT violating state has a negative eigenvalue which reduces the value of $p_3$, in comparison with the purity $p_2$. This effect is mathematically captured by the $p_3$-PPT condition and allowed us to detect PPT violation and thus entanglement for many-body mixed states [see Fig.~\ref{fig:setup}c)]. 
In the SM \cite{SM}, we present additional simulations showing  the power  of the $p_3$-PPT condition, in comparison with the negativity and the condition based on purities of nested subsystems. 

The third PT-moment $p_3$ does not only allow to detect mixed-state entanglement. It can also be used to study the dynamics of entanglement in various many-body quantum systems~\cite{Calabrese2012,Chung2014,Coser2014,Wu2019,Wybo2020}.
Here, we analyze the behavior of the dimensionless ratio  $R_3=-\log_2[p_3/\mathrm{Tr}(\rho_{AB}^3)]$, which, as shown in quantum field theory, follows the same universal behavior as the negativity during evolution with a local  Hamiltonian~\cite{Coser2014}. We remark that $R_3$ is however only well-defined for states with $p_3>0$ (Werner states in large dimensions are a counter-example~\cite{SM}). Furthermore, $R_3$ is not an entanglement monotone~\cite{Wybo2020}. It vanishes for all product states, but can still be strictly positive for certain separable states~\cite{Horodecki2009,Wybo2020}.

Fig.~\ref{fig:R3} illustrates the time evolution of $R_3$ for (a) connected  and (b) disconnected  subsystems $AB$, respectively.
The appearing peaks of $R_3$  have been predicted and analyzed for various one-dimensional 
quantum systems subject  to  local interactions~\cite{Coser2014,Alba2019} (and have also been studied in the context of R\'enyi mutual information~\cite{Alba2019a,Somnath2020}).
They can be understood in terms of propagating quasi-particles which describe collective excitations in the system~\cite{Coser2014,Alba2019}.
In this picture, entanglement between two partitions $A$ and $B$  is induced by the presence of entangled pairs of quasi-particles shared between $A$ and $B$.  For each pair, the individual quasi-particles propagate in opposite directions and start to entangle, in the course of the time evolution, partitions that are more and more separated~\cite{Coser2014,Alba2019}.
In particular, for two adjacent partitions (a), $R_3$ increases at early times, which is consistent with the picture of shared  pairs of entangled quasi-particles entering the two partitions immediately. 
After a certain time $R_3$ reaches a maximum and starts to decrease, which can be understood as the time when the quasi-particles start to ‘escape’ the region $AB$. 
For separated partitions (b), the peaks are delayed due to the finite speed of propagation of the quasi particles. In addition, their maximum value is lowered because of the finite life-times of quasi-particles. The latter feature is characteristic 
to chaotic (non-integrable) thermalizing systems~\cite{Alba2019a} and is in our case further enhanced by decoherence.

\begin{figure}
	\includegraphics[width=1.\linewidth]{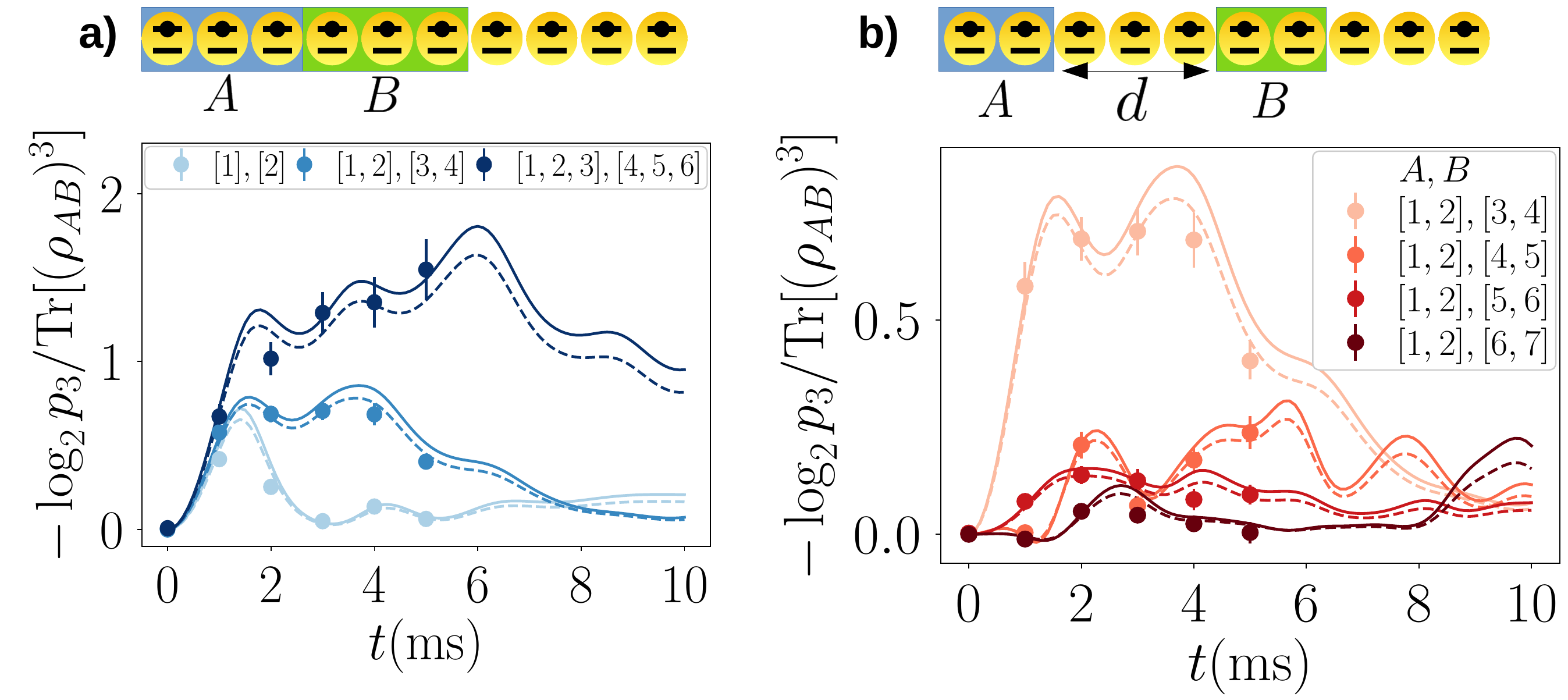}
	\caption{{\it Evolution of the ratio $R_3$ from experimental data~\cite{Brydges2019}}.
		(a) Connected partitions. (b) Disconnected partitions separated by $d=0,1,2,3$ spins.  Different colors correspond to different partitions $AB$.
		Dots are obtained with the shadow estimator Eq.~\eqref{eq:p_n} using experimental data \cite{Brydges2019}. Solid (dashed) lines: theory simulation of  unitary dynamics (including decoherence \cite{Brydges2019}).}
	\label{fig:R3}
\end{figure}

\textit{Conclusion--}
Our protocol extends the paradigm of randomized measurements, yielding the first direct measurement of PT-moments in a many-body system. 
U-statistics provides the key ingredient there and enables us to harness a remarkable advantage over state tomography in terms of statistical errors. At a fundamental level, it is therefore natural to investigate how to access new important physical quantities based on random measurement data, and with significant savings in terms of measurement and classical postprocessing over existing methods.
This approach can be used to derive protocols to directly infer entanglement measures (including non-polynomial functions of the density matrix), such as the von-Neumann entropy and the negativity.

\begin{acknowledgments}
	We are grateful to Alireza Seif who pointed out interesting error scaling effects for classical shadows in a Scirate comment addressing Ref.~\cite{Huang2020}. 
	We thank M.~Knap, S.~Nezami, F.~Pollmann and E.~Wybo for discussions and valuable suggestions, as well as M.~Joshi for the careful reading and comments on the manuscript.
	T.~Brydges, P.~Jurcevic, C.~Maier, B.~Lanyon,  R.~Blatt, and C.~Roos have generously shared the experimental data of Ref.~\cite{Brydges2019}. Simulations were performed with the QuTiP library~\cite{qutip2}.
	Research in Innsbruck is supported by the European Union's Horizon 2020 research and innovation programme under Grant Agreement No.~817482 (PASQuanS) and No.~731473 (QuantERA via QTFLAG), and by the Simons Collaboration on Ultra-Quantum Matter, which is a grant from the Simons Foundation (651440, P.Z.). B. K. acknowledges financial support from the Austrian Academy of Sciences via the Innovation Fund 'Research,
	Science and Society', the SFB BeyondC (Grant No. F7107-N38), and the Austrian Science Fund (FWF) grant DKALM: W1259-N27.
	Research at Caltech is supported by the Kortschak Scholars Program, the US Department of Energy (DE-SC0020290), the US Army Research Office (W911NF-18-1-0103), and the US National Science Foundation (PHY-1733907). The Institute for Quantum Information and Matter is an NSF Physics Frontiers Center. Research in Trieste is partly supported by European Research Council (grant No 758329 and 771536) and by the Italian Ministry of Education under the FARE programme. BV acknowledges funding from the Austrian Science Fundation (FWF, P.~32597N).
\end{acknowledgments}

\bibliographystyle{apsrev4-1}
%

\clearpage
\appendix
\counterwithin{figure}{section}

\section{The $p_3$-PPT condition}

In this section we present, prove and discuss the $p_3$-PPT condition.
The $p_3$-PPT condition is the contrapositive of the following statement about moments of positive semidefinite matrices with unit trace.

\begin{proposition} 
	\label{prop:p3}
	For every positive semidefinite matrix $X$ with unit trace ($\mathrm{Tr}(X)=1$) it holds that
	\begin{equation}
		\mathrm{tr}(X^2)^2 \leq \mathrm{tr}(X^3). \label{eq:p3-condition}
	\end{equation}
\end{proposition}

Note that Eq.~\eqref{eq:p3-condition} resembles the following well-known monotonicity relation among R\'enyi entropies (see e.g., Ref.~\cite{Zyczkowski2003}):
\begin{equation}
	S_{3}(\rho) \leq S_2 (\rho) \; \;\text{for} \; \; S_n (\rho) = \tfrac{1}{1-n}\log_2 \big( \mathrm{tr}(\rho^n) \big).
\end{equation}
However, this relation only applies to density matrices, i.e. positive semidefinite matrices with unit trace. 
The $p_3$-PPT condition, in contrast, is designed to test the absence of positive semidefiniteness. Hence, it is crucial to have a condition that does not break down if the matrix in question has negative eigenvalues. Rel.~\eqref{eq:p3-condition} (and its direct proof provided in the next subsection) do achieve this goal, while an argument based on monotonicity relations between R\'enyi entropies can break down, because the logarithm of non-positive numbers is not properly defined.

\subsection{Proof of the $p_3$-PPT condition}

Let $X$ be a Hermitian $d \times d$ matrix with eigenvalue decomposition $X = \sum_{i=1}^d \lambda_i |x_i \rangle \! \langle x_i|$.
For $p \geq 1$, we
introduce the
Schatten-$p$ norms
\begin{equation*}
	\|X \|_p = \Big( \sum_{i=1}^d |\lambda_i |^p \Big)^{1/p} = \mathrm{Tr}\left( |X|^p  \right)^{1/p},
\end{equation*}
where $|X|=\sqrt{X^2} = \sum_{i=1}^d |\lambda_i | |x_i \rangle \! \langle x_i|$ denotes the (matrix-valued) absolute value.
The Schatten-$p$ norms encompass most widely used matrix norms in quantum information. Concrete examples are the trace norm ($p=1$), the Hilbert-Schmidt/Frobenius norm ($p=2$) and the operator/spectral norm ($p = \infty$). 
Each Schatten-$p$ norm corresponds to the usual vector $\ell_p$-norm of the vector of eigenvalues  $\lambda = (\lambda_1,\ldots,\lambda_d)^T \in \mathbb{R}^d$:
\begin{equation}
	\| \lambda \|_{\ell_p} = \Big(\sum_{i=1}^d |\lambda_i|^p \Big)^{1/p} \quad \text{for $p \geq 1$.}
\end{equation}
Hence, Schatten-$p$ norms inherit many desirable properties from their vector-norm counterparts. Here, we shall use vector norm relations to derive a relation among Schatten-$p$ norms. It is based on Hoelder's inequality that relates the inner product
\begin{equation}
	\langle v,w \rangle = \sum_{i=1}^d v_i w_i \quad \text{for $v,w \in \mathbb{R}^d$}
\end{equation}
to a combination of $\ell_p$ norms.

\begin{fact}[Hoelder's inequality for vector norms]
	Fix $p,q \geq 1$ such that $1/p +1/q =1$. Then,
	\begin{equation}
		\left| \langle v,w \rangle \right|  \leq \sum_{i=1}^d |v_i w_i| \leq \| v\|_{\ell_p} \|w \|_{\ell_q}
	\end{equation}
	for any $v,w \in \mathbb{R}^d$.
\end{fact}
The well-known Cauchy-Schwarz inequality is a special case of this fact.
Set $p=q=1/2$ to conclude
\begin{equation}
	\left|\langle v,w \rangle \right|  \leq \|v \|_{\ell_2} \|w \|_{\ell_2}
	= \langle v,v \rangle^{1/2}\langle w,w \rangle^{1/2}.
\end{equation}
At the heart of our proof for the $p_3$-PPT condition is a simple relation between Schatten-$p$ norms of orders $p=1,2,3$. 

\begin{lemma} \label{lem:p3-aux}
	The following norm relation holds for every Hermitian matrix $X$:
	\begin{equation*}
		\|X \|_2^4 \leq \|X \|_1 \|X \|_3^3
	\end{equation*}
\end{lemma}

\begin{proof}
	Let $\lambda =(\lambda_1,\ldots,\lambda_d)^T$ be the $d$-dimensional vector of eigenvalues of $X$.
	Apply Hoelder's inequality with $p=3,q=3/2$ to the inner product of this vector of eigenvalues with itself:
	\begin{equation}
		\mathrm{Tr}(X^2) = \langle \lambda, \lambda \rangle \leq \| \lambda \|_{\ell_3} \| \lambda\|_{\ell_{2/3}} = \|X \|_3 \| \lambda \|_{\ell_{3/2}}. \label{eq:p3-aux}
	\end{equation}
	Next, we apply Cauchy-Schwarz to the remaining $\ell_{3/2}$-norm:
	\begin{align*}
		\|\lambda \|_{\ell_{3/2}} =& \big( \sum_{i=1}^d |\lambda_i|^{3/2}\big)^{2/3}= \big( \sum_{i=1}^d |\lambda_i| | \lambda_i|^{1/2}\big)^{2/3} \\
		\leq &\Big( \big( \sum_{i=1}^d |\lambda_i|^2 \big)^{1/2} \big( \sum_{i=1}^d |\lambda_i|^{2/2}\big)^{1/2} \Big)^{2/3} \\
		=& \| \lambda \|_{\ell_2}^{2/3} \| \lambda \|_{\ell_1}^{1/3} 
		= \|X \|_2^{2/3} \|X \|_1^{1/3}.
	\end{align*}
	Inserting this relation into Eq.~\eqref{eq:p3-aux} reveals
	\begin{equation*}
		\|X \|_2^2 \leq \|X \|_2^{2/3} \|X \|_1^{1/3} \|X \|_3
	\end{equation*}
	which is equivalent to the claim (take the 3rd power and divide by $\|X\|_2^2$). 
\end{proof}

Proposition~\ref{prop:p3} is an immediate consequence of Lemma~\ref{lem:p3-aux} and elementary properties of positive semidefinite matrices. Recall that a Hermitian $d \times d$ matrix is positive semidefinite (psd) if every eigenvalue is nonnegative. This in turn ensures $|X|=X$ and, by extension, $\|X\|_p = \mathrm{Tr}(X^p)^{1/p}$ for all $p \geq 1$.

\subsection{Discussion and potential generalizations}
The $p_3$-PPT condition tests the absence of positive semidefiniteness based on moments $\mathrm{Tr}(X^p)$ of order $p=1,2,3$. It is natural to wonder whether higher order moments allow the construction of more refined tests. It is possible to show that every positive semidefinite matrix $X$ with unit trace must obey
\begin{align}
	\mathrm{tr}(X^{p-1})^{p-1} \leq \mathrm{tr} (X^p)^{p-2}\quad \text{for all $p >2$.}
	\label{eq:generalized-condition}
\end{align}
As this is a direct extension of the $p_3$-PPT condition ($p=3$), we omit the proof.
Unfortunately, we found numerically that these direct extensions actually produce \emph{weaker} 
tests for the absence of positive semidefiniteness, i.e.\ there exist matrices $X$ that violate the $p_3$-PPT condition but satisfy Rel.~\eqref{eq:generalized-condition} for higher moments $p \geq 4$. This is not completely surprising, since Rel.~\eqref{eq:generalized-condition} compares (powers of) neighboring matrix moments with order $(p-1)$ and $p$. As $p$ increases, these matrix moments 
suppress contributions of small eigenvalues ever more strongly.  In the case of partially transposed quantum states,  the eigenvalues  are required to sum up to one and must be contained in the interval $[-1/2,1]$ \cite{Rana2013}. Thus, the negative eigenvalues can never dominate the spectrum and 
high matrix moment tests for the existence of negative eigenvalues suffer from suppression effects. 

This observation suggests that powerful tests for negative eigenvalues should involve \emph{all} matrix moments $\mathrm{tr}(X^p)$ up to a certain order $p_{\max}$. 
It is useful to change perspective in order to reason about potential improvments. The $p_3$-PPT condition checks whether the following inequality is true:
\begin{equation}
	F_{3}(X) = -\mathrm{tr}(X^3) + \mathrm{tr}(X^2)^2>0.
	\label{eq:p3-polynomial}
\end{equation}
For matrices $X$ with unit trace, 
we can reinterpret the matrix-valued function $F_{3}(X)$ as a sum of (identical) degree-3 polynomials applied to all eigenvalues $\lambda_1,\ldots,\lambda_d$ of $X$. Set $p_2=\mathrm{tr}(X^2)$ and use $\mathrm{tr}(X)=\sum_{i=1}^d \lambda_i =1$ to conclude
\begin{align}
	F_3 (X)=& - \mathrm{tr}(X^3) + 2p_2  \mathrm{tr}(X^2)- p_2^2 \mathrm{tr}(X) \nonumber\\
	=& \sum_{i=1}^d \left( -\lambda_i^3 + 2 p_2 \lambda_i^2 - p_2^2 \lambda_i \right) \nonumber \\
	=& \sum_{i=1}^d - \lambda_i (\lambda_i -p_2)^2
	=: \sum_{i=1}^d f_3 (\lambda_i).
\end{align}
Note that the polynomial 
\begin{equation}
	f_3(x) = -x (x-p_2)^2 \quad \text{for} \quad x \in \mathbb{R} \label{eq:f3}
\end{equation}
depends on $p_2$ and, by extension, also on the matrix $X$. 
We will come back to this aspect later. For now, we point out that -- regardless of the actual value of $p_2$ -- this polynomial has three interesting properties:
\begin{align}
	\begin{array}{cc}
		f_3 (x) \leq 0 
		& \quad \text{if $x >0$}, \\
		f_3 (0) = 0, & \\
		f_3 (x) >0& \quad \text{if $x<0$}.
	\end{array} \label{eq:constraints}
\end{align}
These properties reflect the behavior of another well-known function -- the \emph{(negated) rectifier  function} (ReLU):
\begin{equation}
	r(-x) = \max \left\{0,-x \right\}=
	\begin{cases} 
		0 & \text{if $x  \geq 0$}, \\
		|x| & \text{if $x <0$}.
	\end{cases}
\end{equation}
See Figure~\ref{fig:ReLU} for a visual comparison.
Applying the (negated) rectifier function to the eigenvalues of $X$ would recover the negativity:
\begin{equation}
	\mathcal{N}(X) = \sum_{ \lambda_i <0} |\lambda_i| 
	= \sum_{i=1}^d r(-\lambda_i).
\end{equation}
Hence, it is instructive to interpret $F_3(X)$ as a polynomial approximation to the (non-analytic) negativity function.

\begin{figure}
	\centering
	\includegraphics[width=0.9\linewidth]{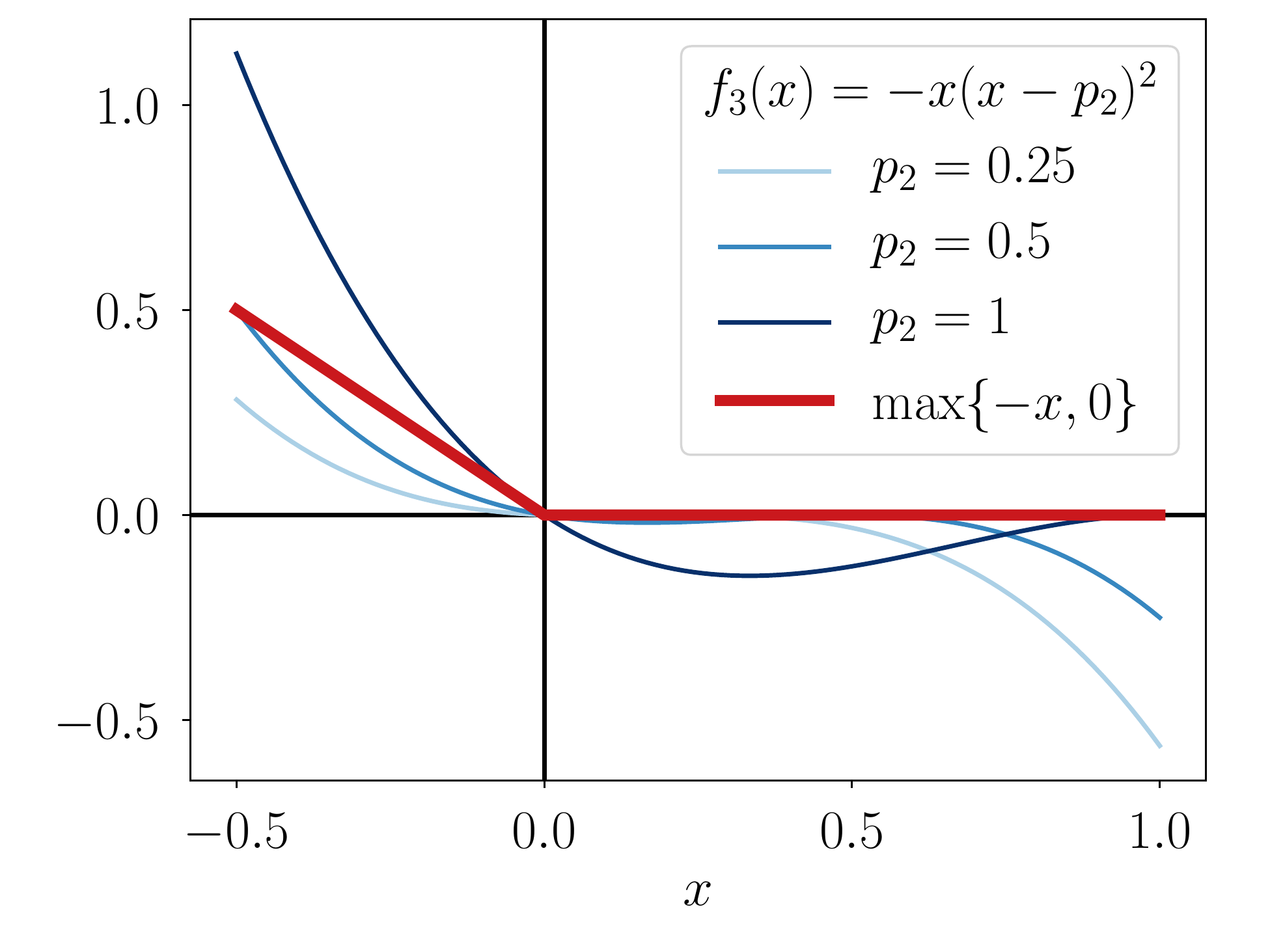}
	\caption{Comparison of $f_3 (x) = -x (x-p_2)^2$ with the negated rectifier function $r(-x)=\max \left\{-x,0\right\}$ for different values of $p_2$ in the relevant interval $[-1/2,1]$ \cite{Rana2013}.}
	\label{fig:ReLU}
\end{figure}

On the level of polynomials, the condition $f_3 (x) \leq 0$ whenever $x>0$ is most important. It implies that positive eigenvalues of $X$ can never increase the value of $F_3 (X)=\sum_{i=1}^d f_3 (\lambda_i)$. In particular, $F_3 (X) \leq 0$ whenever $X$ is positive semidefinite -- as stated in Proposition~\ref{prop:p3}. The $p_3$-PPT condition is sound, i.e.\ it has no false positives. 

Conversely, $f_3(x) >0$ for $x<0$ implies that $F_3 (X)$ can become positive if $X$ has negative eigenvalues. Hence, the $p_3$-PPT condition is not vacuous. It is capable of detecting negative eigenvalues in many, but not all, unit-trace matrices $X$. 

Let us now return to the (matrix-dependent) parameter choice in Eq.~\eqref{eq:f3}. In principle, every polynomial of the form $f_3^{(a)}(x)=-x(x-a)^2$ with $a \in \mathbb{R}$ obeys the important structure constraints \eqref{eq:constraints} and therefore produces a sound test for negative eigenvalues. For fixed $X$, the associated matrix polynomial evaluates to
\begin{equation}
	F_3^{(a)} (X) 
	= - \mathrm{tr}(X^3) + 2 a \mathrm{tr}(X^2)-a^2 \mathrm{tr}(X).
\end{equation}
We can optimize this expression over the parameter $a \in \mathbb{R}$ to make the test as strong as possible. 
The optimal choice is $a_\sharp = \mathrm{tr}(X^2)/\mathrm{tr}(X)$ and produces a matrix polynomial that obeys $F^{(a^\sharp)}_3(X)\geq \max_{a \in \mathbb{R}} F_3^{(a)} (X)$ for $X$ fixed. If $X$ has also unit trace, the optimal parameter becomes $a_\sharp = p_2$ and produces the $p_3$-PPT condition \eqref{eq:p3-polynomial}.

This construction of PPT conditions readily extends to higher order polynomials $f_p(x)=a_p x^p + \cdots a_1 x +a_0$. Increasing the degree $p$ produces more expressive ansatz functions that can approximate the (negated) rectifier function -- and its core properties -- ever more accurately. Viewed from this angle, it becomes apparent that measuring more matrix moments can produce stronger tests for detecting negative eigenvalues. However, it is not so obvious how to choose the parameters $a_p,\ldots,a_0$ ``optimally'',  or what ``optimally'' actually means in this context. Some well-known polynomial approximations of the rectifier function $r(-x)$ -- like Taylor expansions of $s(-x)=\ln (1+\mathrm{e}^{-x})$ (the ``softplus'' function) -- are not well-suited for this task, because $s(-x)>0$ even for $x>0$. This in turn would imply that the associated test condition may not be sound. We believe that a thorough analysis of these questions is timely and interesting, but would go beyond the scope of this work. We intend to address it in future research.

\section{$p_3$-PPT condition for Werner States}

Werner states are bipartite quantum states in a Hilbert space $\mathcal{H}_{AB}=\mathcal{H}_A \otimes \mathcal{H}_B$ with dimensions $d_A=d_B \equiv d$, defined as 
\begin{align}
	\rho_W= \alpha \binom{d+1}{2}^{-1} \Pi_+ + (1-\alpha) \binom{d}{2}^{-1}\Pi_-
	\label{eq:werner}
\end{align}
with parameter $\alpha \in [0,1]$ and $\Pi_\pm=\frac{1}{2}\left( \mathbb{I}\pm \Pi_{12} \right)$ projectors onto symmetric $\mathcal{H}_+$ and anti-symmetric $\mathcal{H}_-$ subspaces of $\mathcal{H}=\mathcal{H}_+ \oplus \mathcal{H}_-$, respectively \cite{Watrous2018}. Here, $\Pi_{12}=\sum_{i,j=1}^d\ket{i}\bra{j} \otimes \ket{j}\bra{i}$ is the swap operator.  We note that the eigenvalues of $\rho_W$ are thus given as $\lambda_+=\alpha  \binom{d+1}{2}^{-1} $ with multiplicity $ \binom{d+1}{2}$ and $\lambda_-=(1-\alpha)\binom{d}{2}^{-1}$ with multiplicity $\binom{d}{2}$. The reduced state $\rho_A$ of qudit $A$  is given by $\rho_A =\tr_B[\rho_W]=\mathbb{I}_A/d$.

Using furthermore that $\Pi^{T_A}_{\pm}= 1/2(\Delta_1 \pm (d\pm 1) \Delta_0)$ with $\Delta_0=\ket{\phi_+}\bra{\phi_+}$ being a projector onto the maximally entangled state and $\Delta_1=\mathbb{I}-\Delta_0$ \cite{Watrous2018}, we find
\begin{align}
	\rho^{T_A}_W= \frac{2\alpha-1 }{d}  \Delta_0 + \frac{1 +d -2\alpha}{d } \frac{\Delta_1}{d^2-1 }
\end{align}
with eigenvalues  $\lambda_0=(2\alpha-1)/{d}$ with multiplicity 1 and  $\lambda_1=({1 +d -2\alpha})/{d (d^2-1) }$ with multiplicity $d^2-1$.

We note that, for any $d$, $\lambda_0<0$ for  $0\leq \alpha < 1/2 $. Thus, using the PPT condition, we find that $\rho_W$ is entangled for $0\leq \alpha < 1/2$. Using the explicit expression of the eigenvalues, we can furthermore determine $\tr\left[ ( \rho^{T_A}) ^{n}\right]$ for any $n$. We find for all local dimensions  $d$
\begin{align}
	\tr\left[ ( \rho^{T_A}) ^2\right]^2 >  \tr\left[ ( \rho^{T_A}) ^3\right]\quad \text{for} \quad
	0\leq \alpha < \frac{1}{2}  
\end{align}
Thus, for Werner states the $p_3$-PPT condition is equivalent to the full PPT condition. It can be furthermore shown that Werner states are separable for $\alpha\geq 1/2$ \cite{Watrous2018}. Thus, for Werner states, the $p_3$-PPT condition is a necessary and sufficient condition for bipartite entanglement.
This also holds true for ``isotropic'' states of the form $\rho = \alpha \mathbbm{1}/d^2 + (1-\alpha) |\phi_+\rangle\langle \phi_+|$, which are closely related.

We 
note that Werner states can have non-positive PT-moments. For local dimension $d>3$ there exists a parameter interval $\left[0,\alpha^*\right)$ such that the associated Werner state \eqref{eq:werner} obeys $p_3=\mathrm{Tr} \big[ (\rho_W^{T_A})^3 \big]<0$ for all $\alpha \in [0,\alpha^*)$.
This 
highlights that the logarithm of PT-moments, appearing also in the ratio $R_3=-\log_2 (p_3/\tr[\rho^3])$,  need not be properly defined, 
justifying a claim from the previous subsection. It 
is difficult to use entropic arguments for reasoning about relations between (logarithmic) PT-moments.

Finally, as shown in Ref.~\cite{Wybo2020}, we remark that $R_3$ is not an entanglement monotone. For separable Werner states with $1/2\leq \alpha<1/2+1/(2d)$, it holds that $0< p_3<\tr[\rho^3]$.  Thus, $R_3=-\log_2 (p_3/\tr[\rho^3])$ can be greater than zero, even for separable states. Since $R_3$  equals zero for all product states, it is not an entanglement monotone \cite{Horodecki2009}.

\section{Comparison of entanglement conditions for quench dynamics}

In this section, we compare the diagnostic power of the full PPT-condition, the $p_3$-PPT condition and a condition based on purities of nested subsystems  to detect bipartite entanglement of mixed states.  Specifically, given a reduced density matrix $\rho_{AB}$ in a bipartite system $AB$, we consider: 
\begin{enumerate}
	\item the   PPT-condition detecting bipartite entanglement between $A$ and $B$ for a strictly positive negativity $\mathcal{N}(\rho_{AB})=\sum_{\lambda<0}|\lambda|>0$, with $\lambda$ the spectrum of $\rho_{AB}^{T_A}$ \cite{Horodecki2009}. 
	\item the $p_3$-PPT condition  detecting bipartite entanglement between $A$ and $B$ for  $1-p_3/p_2^2>0$.
	\item a condition based on the purity of nested subsystems detecting bipartite entanglement between $A$ and $B$ for $\tr[\rho_A^2]<\tr[\rho_{AB}^2]$ with $
	\rho_A=\tr_B[\rho_{AB}]$ the reduced density matrix of subsystem $A$ \cite{Horodecki2009}.
\end{enumerate}
The latter 'purity' condition was used in previous experimental works measuring the second R\'{e}nyi entropy \cite{Islam2015,Kaufman2016,Linke2018, Brydges2019} to reveal bipartite entanglement of weakly mixed states.  

To test these conditions, we consider here, as an example, quantum states generated via quench dynamics in interacting spin models. Specifically, we study quenches in  the $XY$-model with long-range interactions, as defined in Eq.~(6) of the main text, in a total system with $N=10$ spins. The initial separable product state is a N\'{e}el state $\ket{\uparrow \downarrow \uparrow \downarrow \dots}$.  

As shown in Fig.~\ref{fig:EntCritExp}, the negativity (red lines) detects bipartite entanglement for all partitions sizes and all times after the quench.  The $p_3$-PPT condition (blue lines) performs similar for the partitions considered in panel (b) and (c) and is thus able to detect bipartite entanglement for highly mixed states $\rho_{AB}$ whose purity decreases to $~0.3$ for the panel (b) at late times. The $p_3$-PPT conditions fails however to detect the entanglement for the close-to completely mixed states of small partitions $|AB|=4$ at late times, displayed in panel (a). This can be attributed to the fact that the $p_3$-PPT condition only relies on low order PT-moments.
The purity condition (green lines) is only useful for the detection of entanglement  for large partitions $AB$ with $|AB|=8$ (panel (c)). These remain  weakly mixed during the entire time evolution, since the total system of $N=10$ spins is described here by a pure state.

\begin{figure}
	\centering
	\includegraphics[width=\linewidth]{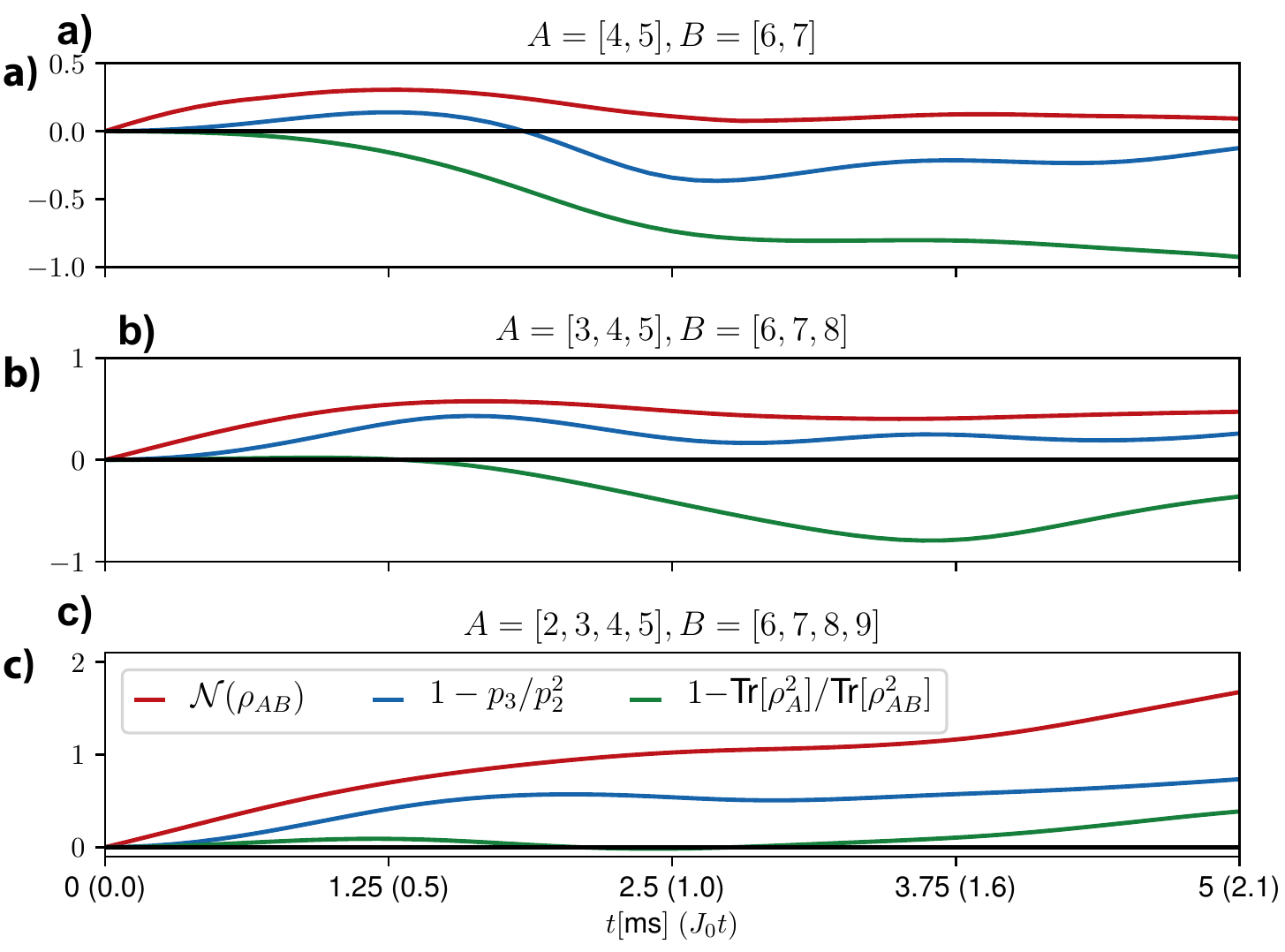}
	\caption{Comparing conditions for bipartite entanglement between two subsystems $A$ and $B$ for states generated with quench dynamics governed by $H_{XY}$ arising from an initial N\'{e}el state in a total system with $N=10$ spins. Modeling the experiment of Ref.~\cite{Brydges2019}, we chose $J_0= 420 s^{-1},\alpha= 1.24$, while other parameter choices lead to 
		similar results.  In all panels, and for all quantities,  a strictly positive value signals bipartite entanglement.}
	\label{fig:EntCritExp}
\end{figure}

\section{Error bars for PT moment predictions} \label{sec:statistics}

Let us first review the data acquisition procedure. To obtain meaningful information about an $N$-qubit state $\rho$, we first perform a collection of random single qubit rotations: $\rho \mapsto u \rho u^\dagger$, where $u=u_{1} \otimes \cdots \otimes u_{N}$ and each $u_{i}$ is chosen from a unitary 3-design.
Subsequently, we perform computational basis measurements and store the outcome:
\begin{align}
	\rho \mapsto u \rho u^\dagger \mapsto |k_1, \dots, k_N \rangle.
\end{align}
Here, $k_1,\ldots,k_N \in \left\{0,1\right\}$ denote the measurement outcomes on qubits $1, \dots, N$. 
As shown in \cite{Elben2018a,Paini2019,Huang2020}, the outcome of this measurement provides a (single-shot) estimate for the unknown state:
\begin{equation}
	\hat{\rho}=\bigotimes_{i =1}^N \left[ 3 (u_{i}^{})^\dag \ket{k_{i}^{}}\bra{k_{i}^{}}u_{i}^{}-\mathbb{I}_2\right] \label{eq:shadow}
\end{equation}
This tensor product is a random matrix -- the unitaries $u^{(i)}$, as well as the observed outcomes $k_i$ are random -- that produces the true underlying state in expectation:
\begin{equation}
	\mathbb{E} \left[ \hat{\rho} \right] = \rho.
\end{equation}
Thus, the result of a (randomly selected) single-shot measurement provides a classical snapshot \eqref{eq:shadow} that reproduces the true underlying state in expectation. This desirable feature extends to density matrices of subsystems. Let $AB \subset \left\{1,\ldots,N \right\}$ be a subset of $|AB|\leq N$ qubits and let $\rho_{AB} = \mathrm{Tr}_{\neg AB} (\rho)$ the associated reduced density matrix. Then,
\begin{align}
	\hat{\rho}_{AB} = \mathrm{Tr}_{\neg AB} (\hat{\rho}) =& \bigotimes_{i \in AB} \left[ 3 (u_{i}^{})^\dag \ket{k_{i}^{}}\bra{k_{i}^{}}u_{i}^{}-\mathbb{I}_2\right] \label{eq:subsystem-shadow} \\
	\text{obeys} & \quad \mathbb{E} \left[ \hat{\rho}_{AB} \right] = \rho_{AB}. \nonumber
\end{align}
This feature can be used to estimate linear properties of the subsystem in question: $o = \mathrm{Tr} \left( O \rho_{AB} \right)$. 
Perform $M$ independent repetitions of the data acquisition procedure and use them to create a collection of (independent) snapshots $\hat{\rho}_{AB}^{(1)},\ldots,\hat{\rho}_{AB}^{(M)}$ -- a ``classical shadow'' \cite{Huang2020} --
and form the empirical average of  subsystem properties:
\begin{equation}
	\hat{o} = \tfrac{1}{M} \sum_{r=1}^M \mathrm{Tr} \left( O \hat{\rho}_{AB}^{(r)} \right).
\end{equation}
Convergence to the target value $o = \mathbb{E} \left[ \hat{o} \right]=\mathrm{Tr}(O \rho_{AB})$ is controlled by the variance. Chebyshev's inequality asserts
\begin{equation}
	\mathrm{Pr} \left[ \left| \hat{o} -o \right| \geq \epsilon \right] \leq \frac{\mathrm{Var} \left[ \hat{o} \right]}{\epsilon^2}
	= \frac{\mathrm{Var} \left[\mathrm{Tr} ( O_{AB} \hat{\rho}_{AB}) \right]}{M \epsilon^2} .\label{eq:chebyshev}
\end{equation}

The remaining (single-shot) variance obeys the following useful relation.

\begin{fact}[Proposition~3 in \cite{Huang2020}] \label{fact:variance}
	Fix a subsystem $AB$ and a linear function $\mathrm{Tr}(O \rho_{AB})$. Then, the single-shot variance associated with $\hat{\rho}_{AB}$ defined in Eq.~\eqref{eq:subsystem-shadow} obeys
	\begin{equation}
		\mathrm{Var} \left[ \mathrm{Tr} \left( O \hat{\rho}_{AB} \right) \right] \leq 2^{|AB|} \mathrm{Tr} \left( O^2 \right). \label{eq:variance}
	\end{equation}
\end{fact}
This inequality is true for any underlying state $\rho$ and bounds the variance in terms of the subsystem dimension $d_A = 2^{|AB|}$ and the Hilbert-Schmidt norm (squared) of the observable $O$. 
Thus, roughly $M \approx 2^{|AB|} \mathrm{Tr}(O^2)/\epsilon^2$ measurement repetitions are necessary to predict $o$ up to accuracy $\epsilon$.

\subsection{Predicting quadratic properties ($p_2$)} \label{sub:quadratic}

The formalism introduced above readily extends to predictions of higher order polynomials. The special case of quadratic functions has already been addressed in Refs.~\cite{VanEnk2012,Elben2018,Brydges2019,Elben2018a}, and Ref.~\cite{Huang2020} (for the present formalism). The key idea is to represent a quadratic function in $\rho$ as a linear function on the tensor product $\rho \otimes \rho$:
\begin{equation}
	o = \mathrm{Tr} \left(O \rho_{AB} \otimes \rho_{AB} \right).
\end{equation}
This function can be approximated by replacing $\rho \otimes \rho$
with a \emph{symmetric} tensor product of two distinct snapshots $\hat{\rho}^{(i)},\hat{\rho}^{(j)}$ ($i \neq j$):
\begin{align}
	\tfrac{1}{2!} &\sum_{\pi \in \mathcal{S}_2} \hat{\rho}_{AB}^{(\pi (i))} \otimes \hat{\rho}_{AB}^{(\pi (j))} \nonumber \\
	&= \tfrac{1}{2} \left( \hat{\rho}_{AB}^{(i)} \otimes \hat{\rho}_{AB}^{(j)} + \hat{\rho}_{AB}^{(j)}\otimes \hat{\rho}_{AB}^{(i)} \right).
\end{align}
There are $\binom{M}{2}$ different ways of combining a collection of $M$ snapshots $\hat{\rho}^{(1)},\ldots,\hat{\rho}^{(M)}$ in this fashion. We can predict $o = \mathrm{Tr} (O \rho_{AB} \otimes \rho_{AB})$ by forming the empirical average over all of them:
\begin{align}
	\hat{o} =& \binom{M}{2}^{-1} \sum_{i < j} \mathrm{Tr} \left( O \tfrac{1}{2!} \sum_{\pi \in \mathcal{S}_2} \hat{\rho}_{AB}^{(\pi (i))} \otimes \hat{\rho}_{AB}^{(\pi (j))} \right) \nonumber \\
	=& \binom{M}{2}^{-1} \sum_{i <j} \mathrm{Tr} \left( O_{(s)} \hat{\rho}_{AB}^{(i)} \otimes \hat{\rho}_{AB}^{(j)} \right).
\end{align}
Here, we have implicitly defined the \emph{symmetrization} $O_{(s)}$ of the original target function $O$.
This ansatz is a special case of Hoeffding's U-statistics estimator \cite{Hoeffding1992}.
Averaging boosts convergence to the desired expectation $\mathbb{E} \left[ \hat{o} \right] = o$ and the speed of convergence is controlled by the variance \eqref{eq:chebyshev}.

Restriction to subsystems is also possible. Suppose that $O$ only acts nontrivially on  a subsystem $AB$ of both state copies. Then,
\begin{equation}
	\hat{o} = \binom{M}{2}^{-1} \sum_{i<j} \mathrm{Tr} \left( O_{(s)} \hat{\rho}_{AB}^{(i)} \otimes \hat{\rho}_{AB}^{(j)} \right) \label{eq:quadratic-estimator}
\end{equation}
and the effective problem dimension becomes $d_{AB}^2 = 4^{|AB|}$. 
The tensor product structure \eqref{eq:shadow} of the individual snapshots allows for generalizing linear variance bounds to this setting. Simply view $\hat{\rho}^{(i)}_{AB} \otimes \hat{\rho}^{(j)}_{AB}$ as a single snapshot of the quantum state $\rho_{AB} \otimes \rho_{AB}$. Fact~\ref{fact:variance} then ensures
\begin{align}
	\mathrm{Var} \left[ \mathrm{Tr} \left( O_{(s)} \hat{\rho}_{AB}^{(i)} \otimes \hat{\rho}_{AB}^{(j)} \right) \right]
	\leq 4^{|AB|} \mathrm{Tr} \left( O_{(s)}^2 \right).
	\label{eq:quadratic-variance}
\end{align}
The full variance of $\hat{o}$ is controlled in part by this relation, but also features linear variance terms \cite[App.~6.A]{Huang2020}. Rather than reviewing this argument in full generality, let us focus on the task at hand: computing the variance associated with predicting the PT-moment of order two.
Fix a bipartite subsystem $AB$ of interest and rewrite $p_2$ as
\begin{align}
	p_2 &= \mathrm{Tr} \left( (\rho_{AB}^{(T_A)})^2 \right) = \mathrm{Tr} \left( \rho_{AB}^2 \right) \nonumber \\
	&= \mathrm{Tr} \left( \Pi_{AB} \rho_{AB} \otimes \rho_{AB} \right).
\end{align}
Here, $\Pi_{AB}$ denotes the swap operator that permutes the entire subsystems $AB$ within two copies of the global system.
We refer to Table~\ref{tab:diagrams} below for a visual derivation of this well-known relation.
The swap operator is symmetric under permuting tensor factors, Hermitian ($\Pi^\dagger_{AB} = \Pi_{AB}$) and orthogonal ($\Pi^2_{AB} = \mathbb{I}_{AB}$).
These properties ensure that the associated general estimator \eqref{eq:quadratic-estimator} can be simplified considerably:
\begin{align}
	\hat{p}_2 =& \binom{M}{2}^{-1} \sum_{i<j} \mathrm{Tr} \left( \Pi_{AB} \hat{\rho}^{(i)}_{AB} \otimes \hat{\rho}^{(j)}_{AB} \right) \nonumber \\
	=& \binom{M}{2}^{-1} \sum_{i<j} \mathrm{Tr} \left( \hat{\rho}^{(i)}_{AB} \hat{\rho}^{(j)}_{AB} \right). \label{eq:p2-estimator}
\end{align}
By construction, $\mathbb{E}\left[ \hat{p}_2 \right] = p_2 = \mathrm{Tr}(\rho^2)$ and the speed of convergence is controlled by the variance. This variance decomposes into a linear and a quadratic part.
We expand the definition of the variance:
\begin{align}
	&  \mathrm{Var} \left[ \hat{p}_2 \right] 
	= \mathbb{E} \left[ (\hat{p}_2 - \mathbb{E}[\hat{p}_2 ] )^2 \right]
	= \mathbb{E} \left[ \hat{p}_2^2 \right] - \mathbb{E} \left[ \hat{p}_2 \right]^2
	\\
	&= \binom{M}{2}^{-2} \sum_{i<j} \sum_{k<l} \Big( \mathrm{Tr} \big( \hat{\rho}_{AB}^{(i)} \hat{\rho}_{AB}^{(j)} \big) \mathrm{Tr} \big( \hat{\rho}_{AB}^{(k)} \hat{\rho}_{AB}^{(l)} \big) - \mathrm{Tr}(\rho_{AB}^2)^2 \Big). \nonumber
\end{align}
The size and nature of each contribution depends on the relation between the indices $i,j,k,l$ \cite{Hoeffding1992}:
\begin{enumerate}
	\item \emph{all indices are distinct:} distinct indices label independent snapshots. In this case the expectation value factorizes completely and produces 
	$\mathbb{E} \left[ \mathrm{Tr}(\hat{\rho}_{AB}^{(i)} \hat{\rho}_{AB}^{(j)}) \mathrm{Tr}(\hat{\rho}_{AB}^{(k)} \hat{\rho}_{AB}^{(l)} ) \right]=\mathrm{Tr}(\rho_{AB}^2)^2$. This is completely offset by the subtraction of the expectation value squared. Hence, terms where all indices are distinct do not contribute to the variance.
	\item \emph{exactly two indices coincide:}
	In this case, the expectation value partly factorizes, e.g.\ $\mathbb{E} \left[ \mathrm{Tr} ( \hat{\rho}_{AB}^{(i)} \hat{\rho}_{AB}^{(j)} ) \mathrm{Tr} ( \hat{\rho}_{AB}^{(k)} \hat{\rho}_{AB}^{(j)} ) \right]
	= \mathbb{E} \left[ \mathrm{Tr} (\rho_{AB} \hat{\rho}_{AB}^{(j)})^2\right]$ for $i \neq k$ and $j=l$. 
	Such index combinations produce a linear variance term $\mathrm{Var} \left[ \mathrm{Tr} ( O \hat{\rho}) \right]$ with $O=\rho_{AB}$. The entire sum contains $\binom{M}{2} \binom{2}{1}\binom{M-2}{2-1}=\binom{M}{2}2(M-2)$ terms of this form.
	\item \emph{two pairs of indices coincide:} there are $\binom{M}{2}\binom{2}{2}\binom{M-2}{2-2}=\binom{M}{2}$ contributions of this form and each of them produces a quadratic variance $\mathrm{Var} \left[ \mathrm{Tr} \left( O \hat{\rho}_{AB} \otimes \hat{\rho}_{AB}'\right) \right]$ with $O = \Pi_{AB}$ (swap).
\end{enumerate}
We conclude that the variance of $\hat{p}_2$ decomposes into linear and quadratic terms. These can be controlled via Rel.~\eqref{eq:variance} and Rel.~\eqref{eq:quadratic-variance}, respectively:
\begin{align}
	\mathrm{Var} \left[ \hat{p}_2 \right] = & \nonumber  \binom{M}{2}^{-1} 
	\left( \vphantom{\left[\hat{\rho}_{AB}^{(2)}\right]} 2 (M-2) \mathrm{Var} \left[ \mathrm{Tr} (\rho_{AB} \hat{\rho}_{AB}) \right] \right.\nonumber  \\
	&\qquad\qquad \left.+ \mathrm{Var} \left[ \mathrm{Tr} ( \Pi_{AB} \hat{\rho}_{AB}^{(1)} \otimes \hat{\rho}_{AB}^{(2)} )\right] \right) \nonumber \\
	\leq & \frac{4(M-2)2^{|AB|}}{M(M-1)}  \mathrm{Tr} \left( \rho_{AB}^2 \right)
	+ \frac{2\times 4^{|AB|}}{M(M-1)} \mathrm{Tr} \left( \Pi^2_{AB} \right) \nonumber\\
	\leq & 4\left(\frac{2^{|AB|} p_2}{M} \right)
	+ 4 \left( \frac{2^{1.5|AB|}}{M} \right)^2. \label{eq:p2-variance}
\end{align}
Chebyshev's inequality \eqref{eq:chebyshev} allows us to translate this insight into an error bound.

\begin{lemma}[Error bound for estimating $p_2$] \label{lem:p2}
	Fix a subsystem $AB$ of interest and suppose that we wish to estimate $p_2 = \mathrm{Tr} \big( (\rho_{AB}^{T_A})^2 \big)$. For $\epsilon, \delta >0$, a total of 
	\begin{equation}
		M \geq 8 \max \left\{ \frac{2^{|AB|} p_2}{\epsilon^2 \delta}, \frac{2^{1.5 |AB|}}{\epsilon \sqrt{\delta}} \right\} 
		\label{eq:p2-sampling-rate}
	\end{equation}
	snapshots suffice to ensure that the estimator \eqref{eq:p2-estimator} obeys $|\hat{p}_2 - p_2 | \leq \epsilon$ with probability at least $1-\delta$.
\end{lemma}

It is worthwhile to briefly discuss this two-pronged error bound. Asymptotically, i.e.\ for $M \to \infty$, the approximation error decays at a rate proportional to $1/\sqrt{M}$. This is the expected asymptotic decay rate for an estimation procedure that relies on empirical averaging (Monte Carlo). The actual rate is also multiplicative, i.e.\ the approximation error is proportional to the target $p_2$.
In the practically more relevant, non-asymptotic setting, things can look strikingly different. For small and moderate sample sizes $M$, the variance bound \eqref{eq:p2-variance} is dominated by the next-to-leading order term ($2^{1.5|AB|} > 2^{|AB|} p_2$, especially if $p_2$ is small). 
Lemma~\ref{lem:p2} captures this discrepancy and heralds an error decay rate proportional to $1/M$ in this regime.

Finally, we point out that the dependence on $\delta$ in Eq.~\eqref{eq:p2-sampling-rate} can be considerably improved by using median of means estimation \cite{Huang2020}: split the total data into equally sized chunks, construct independent estimators and take their median. For this procedure, a sampling rate proportional to $\log (1/\delta)$ suffices. Moreover, median of means
is much more robust towards outlier corruption and allows for using the same data to predict purities of many different subsystems simultaneously. 
This, however, comes at the price of 
somewhat larger constants in the error bound \eqref{eq:p2-sampling-rate} and heralds a tradeoff. In statistical terms, 
median of means estimation dramatically increases confidence levels ($1-\delta$) at the cost of slightly larger
error bars (confidence intervals).
This tradeoff becomes advantageous when one attempts to predict very many properties from a single data set.

\subsection{Predicting cubic properties ($p_3$ and $\mathrm{Tr}(\rho_{AB}^3)$)} \label{sub:cubic}

Cubic properties can be predicted in much the same fashion as quadratic properties \cite{Huang2020}. Write $o = \mathrm{Tr} \left( O \rho_{AB} \otimes \rho_{AB} \otimes \rho_{AB} \right)$ and approximate $\rho \otimes \rho \otimes \rho$ by a \emph{symmetric} tensor product of three distinct snapshots $\hat{\rho}_{AB}^{(i)},\hat{\rho}_{AB}^{(j)},\hat{\rho}_{AB}^{(k)}$:
\begin{equation}
	\tfrac{1}{3!} \sum_{\pi \in \mathcal{S}_3} \hat{\rho}_{AB}^{(\pi (i))} \otimes \hat{\rho}_{AB}^{(\pi (j))} \otimes \hat{\rho}_{AB}^{(\pi (k))}.
\end{equation}
There are $\binom{M}{3}$ different ways of combining a collection of $M$ (independent) snapshots $\hat{\rho}_{AB}^{(1)},\ldots,\hat{\rho}_{AB}^{(M)}$ in this fashion. We estimate the cubic function $o$ by averaging over all of them (U-statistics \cite{Hoeffding1992}):
\begin{align}
	\hat{o} = {\binom{M}{3}}^{-1} \sum_{i<j<k} \mathrm{Tr} \left( O \tfrac{1}{3!} \sum_{\pi \in \mathcal{S}_3} \hat{\rho}^{(\pi (i))} \otimes \hat{\rho}^{(\pi (j))} \otimes \hat{\rho}^{(\pi (k))}\right).
\end{align}
Once more, the variance controls the rate of 
convergence to the desired target value $\mathbb{E} \left[ \hat{o} \right] =\mathrm{Tr} \left( O \rho \otimes \rho \otimes \rho \right)$. This variance decomposes into a linear, a quadratic and a cubic part. The argument is a straightforward generalization of the analysis from the previous subsection. 
Rather than repeating the steps in full generality, we directly focus on the 3rd order PT-moment $p_3$ of a subsystem $AB$:
\begin{equation}
	p_3 = \mathrm{Tr} \left( (\rho_{AB}^{T_A})^3 \right).
\end{equation}
For notational simplicity, we suppress the subscript $AB$ indicating the subsystem of interest and label the shadows by lower-case indices: $\hat{\rho}^{(i)}_{AB} \mapsto \hat{\rho}_i$ for $i=1,\ldots,M$. Due to the cyclicity of the trace, the U-statistics estimator simplifies to
\begin{align}
	\binom{M}{3}\hat{p}_3 =& 
	\sum_{i <j<k} \mathrm{Tr} \left( \tfrac{1}{3!} \sum_{\pi \in \mathcal{S}_3} \hat{\rho}_{\pi (i)}^{T_A} \hat{\rho}_{\pi (j)}^{T_A} \hat{\rho}_{\pi (k)}^{T_A} \right) \label{eq:p3-estimator} \\
	=&
	\sum_{i<j<k} \tfrac{1}{2} \left( \mathrm{Tr} \left( \hat{\rho}_i^{T_A} \hat{\rho}_j^{T_A} \hat{\rho}_k^{T_A} \right) + \mathrm{Tr} \left( \hat{\rho}_j^{T_A} \hat{\rho}_i^{T_A} \hat{\rho}_k^{T_A} \right) \right), \nonumber
\end{align}
where we have moved the normalization factor $\binom{M}{3}^{-1}$ to the left hand side in order to
to increase readability. 
When computing the variance, we need to consider two sums over triples of distinct indices in $\left\{1,\ldots,M\right\}$. If all indices are distinct, the overall contribution vanishes. Otherwise the contribution depends on the number $c \in \left\{1,2,3\right\}$ of indices the triples have in common. The number of distinct choices for two triples with exactly $c$ integers in common is $\binom{M}{3}\binom{3}{c} \binom{M-3}{3-c}$ and we infer
\begin{align}
	& \binom{M}{3} \mathrm{Var} \left[ \hat{p}_3 \right] \nonumber \\
	=& 
	\binom{3}{1} \binom{M-3}{2}
	\mathrm{Var} \left[ \mathrm{Tr} \left( (\rho^{T_A})^2 \hat{\rho}^{T_A}_1\right) \right] \nonumber\\
	+& 
	\binom{3}{2} \binom{M-3}{1}
	\mathrm{Var} \left[\mathrm{Tr} \left( \rho^{T_A} \tfrac{1}{2} \left( \hat{\rho}_1^{T_A} \rho_2^{T_A}+\hat{\rho}_2^{T_A} \hat{\rho}_1^{T_A} \right) \right) \right] \nonumber\\
	+& 
	\mathrm{Var} \left[ \tfrac{1}{2} \left( \mathrm{Tr} \left( \hat{\rho}_1^{T_A} \hat{\rho}_2^{T_A} \hat{\rho}_3^{T_A} \right) + \mathrm{Tr} \left( \hat{\rho}_2^{T_A} \hat{\rho}_1^{T_A} \hat{\rho}_3^{T_A} \right) \right) \right] \nonumber \\
	\leq & \binom{M}{3}\left( \frac{9}{M} L + \frac{18}{M^2} Q + \frac{12}{M^3} C \right).
\end{align}
Here, $\hat{\rho}_1,\hat{\rho}_2,\hat{\rho}_3$ denote independent, random realizations of the snapshot $\hat{\rho}$ and we have introduced place-holders for linear ($L$), quadratic ($Q$) and cubic ($C$) contributions, respectively.
For the task at hand, these contributions can be bounded individually and depend on the subsystem size $AB$:
\begin{enumerate}
	\item \emph{linear contribution:} set $O = (\rho^{T_A}_{AB})^2$ for notational brevity.
	We can use $\mathrm{Tr}(O \hat{\rho}^{T_A}) = \mathrm{Tr}(O^{T_A} \hat{\rho} )$ to absorb the partial transpose in the linear function. Rel.~\eqref{eq:variance} then ensures
	\begin{align}
		L \leq 2^{|AB|} \mathrm{Tr} (\rho^2)^2,
	\end{align}
	where we have also used $\mathrm{Tr}((O^{T_A})^2) = \mathrm{Tr}(O^2)$, as well as $\mathrm{Tr}(O^2) = \|O \|_2^2 \leq \|O\|_1^2= \mathrm{Tr}(O)^2$, because $O=\rho^2$ is psd.
	\item \emph{quadratic contribution:} We can bring
	$ \tfrac{1}{2} \left( \mathrm{Tr}(\rho^{T_A} \hat{\rho}_1^{T_A} \hat{\rho}_2^{T_A}) + \mathrm{Tr} (\rho^{T_A} \hat{\rho}_2^{T_A} \hat{\rho}_1^{T_A} \right)$ 
	into the canonical form $\mathrm{Tr} \left( O \hat{\rho}_{AB}^{(1)} \otimes \hat{\rho}_{AB}^{(2)} \right)$ by introducing
	\begin{align}
		O = \tfrac{1}{2} &\left( \Pi_A (\rho \otimes \mathbb{I}_{AB}) \Pi_B \right. \nonumber \\ & \left.  + \Pi_B (\rho \otimes \mathbb{I}_{AB}) \Pi_A
		\right).
	\end{align}
	We refer to Table~\ref{tab:diagrams} for a visual derivation.
	Here, $\Pi_A$ and $\Pi_B$ are permutation operators that swap the two $A$- and $B$-systems, respectively. Rel.~\eqref{eq:quadratic-variance} then ensures
	\begin{align}
		Q \leq 2^{2|AB|} \mathrm{Tr} \left( O^2 \right) \leq 
		2^{3|AB|} \mathrm{Tr}(\rho^2).
	\end{align}
	The final estimate follows from exploiting $\Pi_A^2 = \Pi_B^2=\mathbb{I}_{AB}$, as well as $\mathrm{Tr} \left( \rho^2 \otimes \mathbb{I}_{AB}^2\right) = 2^{|AB|} \mathrm{Tr}(\rho^2)$.
	\item \emph{cubic contribution:} We can bring the cubic function $\tfrac{1}{2} \big( \mathrm{Tr}(\hat{\rho}^{T_A}_1 \hat{\rho}^{T_A}_2 \hat{\rho}^{T_A}_3 ) + \mathrm{Tr} (\hat{\rho}^{T_A}_2 \hat{\rho}^{T_A}_1 \hat{\rho}^{T_A}_3) \big)$
	into the canonical form $\mathrm{Tr}\big(O \hat{\rho}^{}_1 \otimes \hat{\rho}^{}_2 \otimes \hat{\rho}^{}_3 \big)$ by introducing
	\begin{align}
		O=\tfrac{1}{2} \left( \overrightarrow{\Pi}_A \otimes \overleftarrow{\Pi}_B + \overrightarrow{\Pi}^\dagger_{A} \otimes \overleftarrow{\Pi}^\dagger_{B} \right),
	\end{align}
	see Table~\ref{tab:diagrams} below.
	Here, $\overrightarrow{\Pi}_A$ is a cyclic permutation that exchanges all $A$-systems in a ``forward'' fashion ($A_1 \mapsto A_2$, $A_2 \mapsto A_3$, $A_3 \mapsto A_1$), while $\overleftarrow{\Pi}_B$ is another cyclic permutation that exchanges all $B$-systems in a ``backwards'' fashion ($B_3 \mapsto B_2$, $B_2 \mapsto B_1$, $B_1 \mapsto B_3$). A staightforward extension of Rel.~\eqref{eq:quadratic-variance} to cubic functions implies
	\begin{align}
		C \leq 2^{3|AB|} \mathrm{Tr}(O^2) \leq 2^{6|AB|},
	\end{align}
	because permutations are orthogonal ($\Pi \Pi^\dagger =\mathbb{I}$) and $\mathrm{Tr}(O^2)$ is dominated by $\mathrm{Tr}(\mathbb{I}_{AB} \otimes \mathbb{I}_{AB} \otimes \mathbb{I}_{AB}) =2^{3|AB|}$.
\end{enumerate}
Inserting these bounds into the variance formula for $p_3$ reveals
\begin{align}
	\mathrm{Var} \left[ \hat{p}_3 \right] \leq & \frac{9}{M} L + \frac{18}{M^2} Q + \frac{12}{M^3} C  \\
	\leq & 9\frac{2^{|AB|}}{M} \mathrm{Tr}(\rho^2)^2 + 18\frac{2^{3|AB|}}{M^2} \mathrm{Tr}(\rho^2) + 12\frac{2^{6|AB|}}{M^3} . \nonumber 
\end{align}
Combining this insight with Chebyshev's inequality \eqref{eq:chebyshev} produces a suitable error bound. Recall that $p_2 = \mathrm{Tr}\big ( (\rho^{T_A})^2 \big) = \mathrm{Tr}(\rho^2) \in [2^{-|AB|},1]$ denotes the purity of the subsystem in question.

\begin{lemma}[Error bound for estimating $p_3$] \label{lem:p3}
	Fix a subsystem $AB$ of interest and suppose that we wish to estimate $p_3 = \mathrm{Tr} \big( (\rho^{T_A})^3 \big)$. For $\epsilon, \delta >0$, a total of 
	\begin{equation}
		M \geq 39 \max \left\{ \frac{2^{|AB|} p_2^2}{\epsilon^2 \delta}, \frac{2^{1.5 |AB|}p_2}{\epsilon \sqrt{\delta}}, \frac{2^{2|AB|}}{\epsilon^{2/3} \delta^{1/3}} \right\} 
		\label{eq:p3-sampling-rate}
	\end{equation}
	snapshots suffice to ensure that the estimator \eqref{eq:p3-estimator} obeys $|\hat{p}_3 - p_3 | \leq \epsilon$ with probability at least $1-\delta$. 
\end{lemma}

This bound on the sampling rate provides different error decay rates for different regimes. For $M \to \infty$, the first term in the maximum dominates and the error decays at an asymptotically unavoidable rate proportional to $1/\sqrt{M}$. 
Conversely, for very small sample sizes $M$, the third term dominates and conveys a much larger decay rate proportional to $1/M^{3/2}$. In the intermediate regime, the second term may dominate and lead to a inverse linear decay rate $1/M$, instead.
The dependence on the error parameter $\delta$ can once more be considerably improved (from $1/\delta$ to $\log (1/\delta)$) by using median of means estimation. This refinement also allows for using the same data to predict the cubic PT-moment of very many subsystems simultaneously \cite{Huang2020}. 

Finally, we point out that the estimation error for $s_3 =\|\rho \|_3^3=\mathrm{Tr}(\rho^3)$ can be bounded in exactly the same fashion. 
For $\epsilon,\delta >0$, a sampling rate $M$ that obeys Rel.~\eqref{eq:p3-sampling-rate} also ensures that the U-statistics estimator U-statistics estimator
\begin{align}
	\hat{s}_3 = \binom{M}{3}^{-1} \sum_{i<j<k} \nonumber\tfrac{1}{2} &\left( \mathrm{Tr} \left( \hat{\rho}^{}_i \hat{\rho}^{}_j \hat{\rho}^{}_k \right) \right. \\
	&\left.\; + \mathrm{Tr} \left( \hat{\rho}^{}_j \hat{\rho}^{}_i \hat{\rho}^{}_k \right) \right)
\end{align}
obeys $\left| \hat{s}_3 - s_3 \right| \leq \epsilon$ with probability $1-\delta$.

The proof is almost identical to the $p_3$-analysis and we leave it as an exercise for the dedicated reader. 

\subsection{Additional numerical simulations}

Here, we complement Fig.~2 of the MT by showing in Fig.~\ref{fig:staterrorsSM} statistical errors in the estimation of $p_2$ and $p_3$ for the ground state of the transverse Ising model $H=J(\sum_i \sigma_i^x  \sigma_{i+1}^x + \sigma_i^z)$ at criticality. We observe the same scaling behavior as in the case of the GHZ state.  For $p_2$ [panel a)], there are indeed two regimes with different decay rates ($1/M$ and $1/\sqrt{M}$). For  $p_3$ [panel b)],
the latter two decay rates $1/M$ and $1/\sqrt{M}$ are also clearly visible. In contrast, the early regime decay rate is not as pronounced.
This is likely due to limited system sizes -- $1/M^{3/2}$ does appropriately capture the decay of red dots (largest system size considered) in the top left corner, but seems to be absent in decay rates for smaller system sizes.

\begin{figure}
	\includegraphics[width=1.\linewidth]{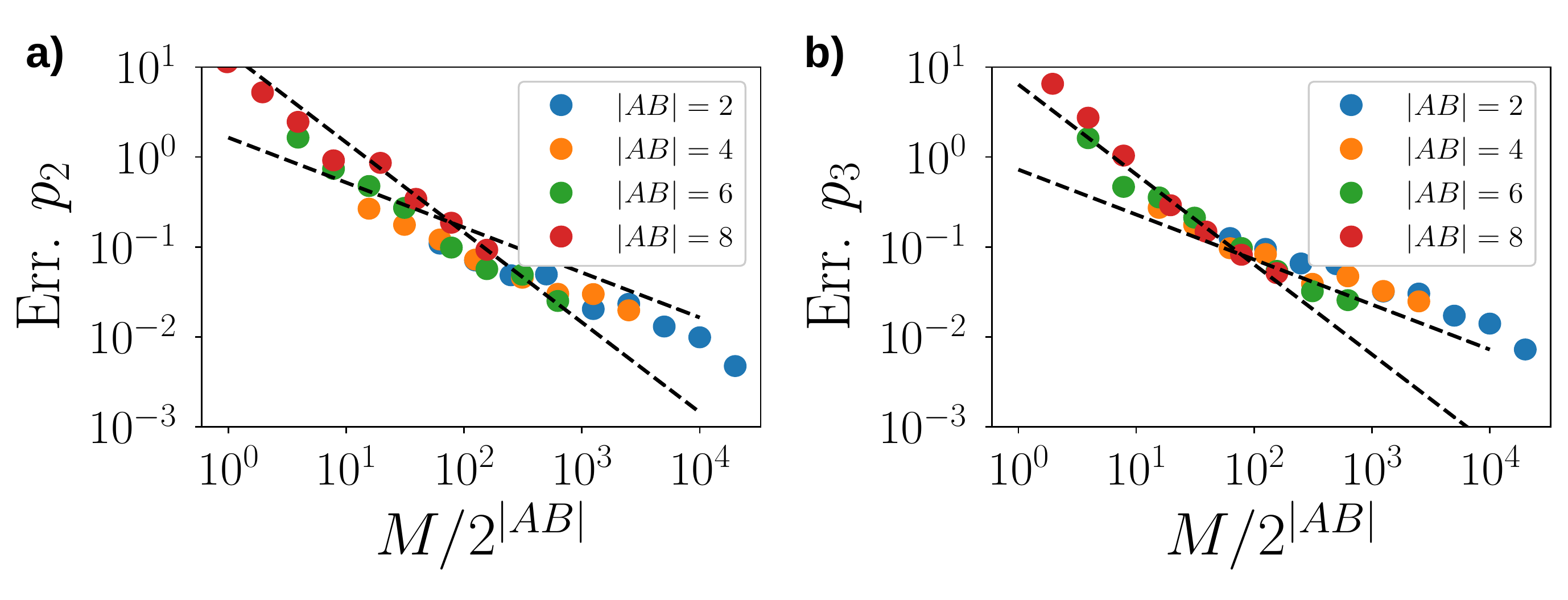}
	\caption{Statistical errors for the ground state of the transverse field Ising model. Dashed lines represent scalings of  $\propto 1/M$, and $\propto1/\sqrt{M}$. In all cases, the number of measurements to
		estimate $p_2$ a) and $p_3$ b) with accuracy $~0.1$ is of the order of  $100\times 2^{|AB|}$.}
	\label{fig:staterrorsSM}
\end{figure}

\section{Auxiliary results and wiring diagrams} \label{sec:wiring}

\begin{table*}[t!]
	\resizebox{2\columnwidth}{!}{
		\begin{tabular}{|c|c|c|c|}
			\hline
			\textbf{expression} & \textbf{diagram representation} & \textbf{diagram reformulation } & \textbf{modified expression} \\
			\hline
			\hline
			\begin{tikzpicture}[baseline,scale=0.5]
				\draw[transparent] (0,-1) -- (0,2);
				\node at (0,0.5) {
					$ \mathrm{Tr} ( X^{T_A}_{AB} Y^{T_A}_{AB})$
				};
			\end{tikzpicture}
			&
			\begin{tikzpicture}[baseline,scale=0.5]
				\draw[transparent] (1.5,-1) -- (1.5,2);
				\draw[thick] (-1.5,0) -- (4.5,0);
				\draw[rounded corners,fill=white] (0-0.75,-0.25) rectangle (0+0.75,1.25);
				\node at (0,0.5) {$X_{AB}$};
				\draw[rounded corners,fill=white] (3-0.75,-0.25) rectangle (3+0.75,1.25);
				\node at (3,0.5) {$Y_{AB}$};
				\draw[thick, rounded corners] (-0.75,1) -- (-1,1) -- (-1,1.5) -- (1.25,1.5) -- (1.25,1) -- (1.5,1);
				\draw[thick, rounded corners] (0.75,1) -- (1,1) -- (1,1.75) -- (-1.25,1.75) -- (-1.25,1) -- (-1.5,1);
				\draw[thick, rounded corners] (2.25,1) -- (2,1) --(2,1.5) -- (4.25,1.5) -- (4.25,1) -- (4.5,1);
				\draw[thick, rounded corners] (3.75,1) -- (4,1) -- (4,1.75) -- (1.75,1.75) -- (1.75,1) -- (1.5,1);
				\foreach \x in {-1.75,4.75}
				{
					\draw (\x,0) circle (0.25);
					\draw (\x,1) circle (0.25);
				}
			\end{tikzpicture}
			&
			\begin{tikzpicture}[baseline,scale=0.5]
				\draw[transparent] (1.5,-1) -- (1.5,2);
				\draw[thick] (-1.5,0) -- (4.5,0);
				\draw[thick] (-1.5,1) -- (4.5,1);
				\draw[rounded corners,fill=white] (0-0.75,-0.25) rectangle (0+0.75,1.25);
				\node at (0,0.5) {$X_{AB}$};
				\draw[rounded corners,fill=white] (3-0.75,-0.25) rectangle (3+0.75,1.25);
				\node at (3,0.5) {$Y_{AB}$};
				\foreach \x in {-1.75,4.75}
				{
					\draw (\x,0) circle (0.25);
					\draw (\x,1) circle (0.25);
				}
			\end{tikzpicture}
			& 
			\begin{tikzpicture}[baseline,scale=0.5]
				\draw[transparent] (0,-1) -- (0,2);
				\node at (0,0.5) {
					$\mathrm{Tr} \left( X_{AB} Y_{AB} \right)$
				};
			\end{tikzpicture}
			\\
			\hline
			\begin{tikzpicture}[baseline,scale=0.5]
				\draw[transparent] (0,-0.5) -- (0,3.5);
				\node at (0,1.5) {
					$\mathrm{Tr} \left( X_{AB} Y_{AB} \right)$
				};
			\end{tikzpicture}
			&
			\begin{tikzpicture}[baseline,scale=0.5]
				\draw[transparent] (0,-0.5) -- (0,3.5);
				\begin{scope}[yshift=1cm]
					\draw[thick] (-1.5,0) -- (4.5,0);
					\draw[thick] (-1.5,1) -- (4.5,1);
					\draw[rounded corners,fill=white] (0-0.75,-0.25) rectangle (0+0.75,1.25);
					\node at (0,0.5) {$X_{AB}$};
					\draw[rounded corners,fill=white] (3-0.75,-0.25) rectangle (3+0.75,1.25);
					\node at (3,0.5) {$Y_{AB}$};
					\foreach \x in {-1.75,4.75}
					{
						\draw (\x,0) circle (0.25);
						\draw (\x,1) circle (0.25);
					}
				\end{scope}
			\end{tikzpicture}
			&
			\begin{tikzpicture}[baseline,scale=0.5]
				\draw[transparent] (0,-0.5) -- (0,3.5);
				\foreach \y in {0,1,2,3}
				{
					\draw[thick] (-1,\y) -- (1,\y);
					\draw[thick] (-3,\y) -- (-3.5,\y);
					\draw (-3.75,\y) circle (0.25);
					\draw[thick] (3,\y) -- (5.5,\y);
					\draw (5.75,\y) circle (0.25);
				}
				\draw[thick, rounded corners] (-1,2) -- (-1.5,2) -- (-2.5,0) -- (-3,0);
				\draw[thick, rounded corners] (-1,0) -- (-1.5,0) -- (-2.5,2) -- (-3,2);
				\draw[thick] (-1,1) -- (-3,1);
				\draw[thick] (-1,3) -- (-3,3);
				\draw[Blue,dashed,rounded corners] (-3,-0.25) rectangle (-1,3.25);
				\fill[Blue,opacity=0.2,rounded corners] (-3,-0.25) rectangle (-1,3.25);
				\node at (-2,2.5) {\textcolor{Blue}{$\Pi_B$}};
				\draw[thick, rounded corners] (1,3) -- (1.5,3) -- (2.5,1) -- (3,1);
				\draw[thick, rounded corners] (1,1) -- (1.5,1) -- (2.5,3) -- (3,3);
				\draw[thick] (1,0) -- (3,0);
				\draw[thick] (1,2) -- (3,2);
				\draw[Red,dashed,rounded corners] (1,-0.25) rectangle (3,3.25);
				\fill[Red,opacity=0.2,rounded corners] (1,-0.25) rectangle (3,3.25);
				\node at (2,0.5) {\textcolor{Red}{$\Pi_A$}};
				\draw[rounded corners,fill=white] (4-0.75,-0.25) rectangle (4+0.75,1.25);
				\node at (4,0.5) {$Y_{AB}$};
				\draw[rounded corners,fill=white] (4-0.75,2-0.25) rectangle (4+0.75,3.25);
				\node at (4,2.5) {$X_{AB}$};
			\end{tikzpicture}
			&
			\begin{tikzpicture}[baseline,scale=0.5]
				\draw[transparent] (0,-0.5) -- (0,3.5);
				\node at (0,1.5) {
					$\mathrm{Tr} \left( \textcolor{Blue}{\Pi_B} \textcolor{Red}{\Pi_A} X_{AB} \otimes Y_{AB} \right)$
				};
				\node at (0,0) {\footnotesize $\textcolor{Blue}{\Pi_B}, \textcolor{Red}{\Pi_A}$: swaps};
			\end{tikzpicture}
			\\
			\hline
			\begin{tikzpicture}[baseline,scale=0.5]
				\draw[transparent] (0,-0.5) -- (0,3.5);
				\node at (0,1.5) {
					$ \mathrm{Tr} ( \rho^{T_A}_{AB} X^{T_A}_{AB} Y^{T_A}_{AB})$
				};
			\end{tikzpicture}
			&
			\begin{tikzpicture}[baseline,scale=0.5]
				\draw[transparent] (0,-0.5) -- (0,3.5);
				\begin{scope}[yshift=1cm]
					\draw[thick] (-1.5,0) -- (7.5,0);
					\foreach \x in {0,3,6}
					{
						\draw[rounded corners,fill=white] (\x-0.75,-0.25) rectangle (\x+0.75,1.25);
						\draw[thick, rounded corners] (\x-0.75,1) -- (\x-1,1) -- (\x-1,1.5) -- (\x+1.25,1.5) -- (\x+1.25,1) -- (\x+1.5,1);
						\draw[thick, rounded corners] (\x+0.75,1) -- (\x+1,1) -- (\x+1,1.75) -- (\x-1.25,1.75) -- (\x-1.25,1) -- (\x-1.5,1);
					}
					\node at (0,0.5) {$\rho_{AB}$};
					\node at (3,0.5) {$X_{AB}$};
					\node at (6,0.5) {$Y_{AB}$};
					\foreach \x in {-1.75,7.75}
					{
						\draw (\x,0) circle (0.25);
						\draw (\x,1) circle (0.25);
					}
				\end{scope}
			\end{tikzpicture}
			&
			\begin{tikzpicture}[baseline,scale=0.5]
				\draw[transparent] (0,-0.5) -- (0,3.5);
				\foreach \y in {0,1,2,3}
				{
					\draw[thick] (-1,\y) -- (1,\y);
					\draw[thick] (-3,\y) -- (-3.5,\y);
					\draw (-3.75,\y) circle (0.25);
					\draw[thick] (3,\y) -- (5.5,\y);
					\draw (5.75,\y) circle (0.25);
				}
				\draw[rounded corners,fill=white] (-0.75,1.75) rectangle (0.75,3.25);
				\node at (0,2.5) {$\rho_{AB}$};
				\draw[rounded corners, fill=white,opacity=0.5] (-0.75,-0.25) rectangle (0.75,1.25);
				\node at (0,0.5) {$\mathbb{I}_{AB}$};
				\draw[thick, rounded corners] (-1,2) -- (-1.5,2) -- (-2.5,0) -- (-3,0);
				\draw[thick, rounded corners] (-1,0) -- (-1.5,0) -- (-2.5,2) -- (-3,2);
				\draw[thick] (-1,1) -- (-3,1);
				\draw[thick] (-1,3) -- (-3,3);
				\draw[Blue,dashed,rounded corners] (-3,-0.25) rectangle (-1,3.25);
				\fill[Blue,opacity=0.2,rounded corners] (-3,-0.25) rectangle (-1,3.25);
				\node at (-2,2.5) {\textcolor{Blue}{$\Pi_B$}};
				\draw[thick, rounded corners] (1,3) -- (1.5,3) -- (2.5,1) -- (3,1);
				\draw[thick, rounded corners] (1,1) -- (1.5,1) -- (2.5,3) -- (3,3);
				\draw[thick] (1,0) -- (3,0);
				\draw[thick] (1,2) -- (3,2);
				\draw[Red,dashed,rounded corners] (1,-0.25) rectangle (3,3.25);
				\fill[Red,opacity=0.2,rounded corners] (1,-0.25) rectangle (3,3.25);
				\node at (2,0.5) {\textcolor{Red}{$\Pi_A$}};
				\draw[rounded corners,fill=white] (4-0.75,-0.25) rectangle (4+0.75,1.25);
				\node at (4,0.5) {$Y_{AB}$};
				\draw[rounded corners,fill=white] (4-0.75,2-0.25) rectangle (4+0.75,3.25);
				\node at (4,2.5) {$X_{AB}$};
			\end{tikzpicture}
			&
			\begin{tikzpicture}[baseline,scale=0.5]
				\draw[transparent] (0,-0.5) -- (0,3.5);
				\node at (0,1.5) {
					$\mathrm{Tr} \left( \textcolor{Blue}{\Pi_B} (\rho_{AB} \otimes \textcolor{gray}{\mathbb{I}_{AB}}) \textcolor{Red}{\Pi_A} X_{AB} \otimes Y_{AB} \right)$
				};
			\end{tikzpicture}
			\\
			\hline
			\begin{tikzpicture}[baseline,scale=0.5]
				\draw[transparent] (0,-0.5) -- (0,3.5);
				\node at (0,1.5) {
					$ \mathrm{Tr} ( \rho^{T_A}_{AB} Y^{T_A}_{AB} X^{T_A}_{AB})$
				};
			\end{tikzpicture}
			&
			\begin{tikzpicture}[baseline,scale=0.5]
				\draw[transparent] (0,-0.5) -- (0,3.5);
				\begin{scope}[yshift=1cm]
					\draw[thick] (-1.5,0) -- (7.5,0);
					\foreach \x in {0,3,6}
					{
						\draw[rounded corners,fill=white] (\x-0.75,-0.25) rectangle (\x+0.75,1.25);
						\draw[thick, rounded corners] (\x-0.75,1) -- (\x-1,1) -- (\x-1,1.5) -- (\x+1.25,1.5) -- (\x+1.25,1) -- (\x+1.5,1);
						\draw[thick, rounded corners] (\x+0.75,1) -- (\x+1,1) -- (\x+1,1.75) -- (\x-1.25,1.75) -- (\x-1.25,1) -- (\x-1.5,1);
					}
					\node at (0,0.5) {$\rho_{AB}$};
					\node at (3,0.5) {$Y_{AB}$};
					\node at (6,0.5) {$X_{AB}$};
					\foreach \x in {-1.75,7.75}
					{
						\draw (\x,0) circle (0.25);
						\draw (\x,1) circle (0.25);
					}
				\end{scope}
			\end{tikzpicture}
			&
			\begin{tikzpicture}[baseline,scale=0.5]
				\draw[transparent] (0,-0.5) -- (0,3.5);
				\foreach \y in {0,1,2,3}
				{
					\draw[thick] (-1,\y) -- (1,\y);
					\draw[thick] (-3,\y) -- (-3.5,\y);
					\draw (-3.75,\y) circle (0.25);
					\draw[thick] (3,\y) -- (5.5,\y);
					\draw (5.75,\y) circle (0.25);
				}
				\draw[rounded corners,fill=white] (-0.75,1.75) rectangle (0.75,3.25);
				\node at (0,2.5) {$\rho_{AB}$};
				\draw[rounded corners, fill=white,opacity=0.5] (-0.75,-0.25) rectangle (0.75,1.25);
				\node at (0,0.5) {$\mathbb{I}_{AB}$};
				\draw[thick, rounded corners] (-1,3) -- (-1.5,3) -- (-2.5,1) -- (-3,1);
				\draw[thick, rounded corners] (-1,1) -- (-1.5,1) -- (-2.5,3) -- (-3,3);
				\draw[thick] (-1,0) -- (-3,0);
				\draw[thick] (-1,2) -- (-3,2);
				\draw[Red,dashed,rounded corners] (-3,-0.25) rectangle (-1,3.25);
				\fill[Red,opacity=0.2,rounded corners] (-3,-0.25) rectangle (-1,3.25);
				\node at (-2,0.5) {\textcolor{Red}{$\Pi_A$}};
				\draw[thick, rounded corners] (1,2) -- (1.5,2) -- (2.5,0) -- (3,0);
				\draw[thick, rounded corners] (1,0) -- (1.5,0) -- (2.5,2) -- (3,2);
				\draw[thick] (1,1) -- (3,1);
				\draw[thick] (1,3) -- (3,3);
				\draw[Blue,dashed,rounded corners] (1,-0.25) rectangle (3,3.25);
				\fill[Blue,opacity=0.2,rounded corners] (1,-0.25) rectangle (3,3.25);
				\node at (2,2.5) {\textcolor{Blue}{$\Pi_B$}};
				\draw[rounded corners,fill=white] (4-0.75,-0.25) rectangle (4+0.75,1.25);
				\node at (4,0.5) {$Y_{AB}$};
				\draw[rounded corners,fill=white] (4-0.75,2-0.25) rectangle (4+0.75,3.25);
				\node at (4,2.5) {$X_{AB}$};
			\end{tikzpicture}
			&
			\begin{tikzpicture}[baseline,scale=0.5]
				\draw[transparent] (0,-0.5) -- (0,3.5);
				\node at (0,1.5) {
					$\mathrm{Tr} \left( \textcolor{Red}{\Pi_A} (\rho_{AB} \otimes \textcolor{gray}{\mathbb{I}_{AB}}) \textcolor{Blue}{\Pi_B} X_{AB} \otimes Y_{AB} \right)$
				};
			\end{tikzpicture}
			\\
			\hline
			\begin{tikzpicture}[baseline,scale=0.5]
				\draw[transparent] (0,-0.5) -- (0,5.5);
				\node at (0,2.5) {$\mathrm{Tr} \left( X^{T_A}_{AB} Y^{T_A}_{AB} Z^{T_A}_{AB} \right)$};
			\end{tikzpicture}
			&
			\begin{tikzpicture}[baseline,scale=0.5]
				\draw[transparent] (0,-0.5) -- (0,5.5);
				\begin{scope}[yshift=2cm]
					\draw[thick] (-1.5,0) -- (7.5,0);
					\foreach \x in {0,3,6}
					{
						\draw[rounded corners,fill=white] (\x-0.75,-0.25) rectangle (\x+0.75,1.25);
						\draw[thick, rounded corners] (\x-0.75,1) -- (\x-1,1) -- (\x-1,1.5) -- (\x+1.25,1.5) -- (\x+1.25,1) -- (\x+1.5,1);
						\draw[thick, rounded corners] (\x+0.75,1) -- (\x+1,1) -- (\x+1,1.75) -- (\x-1.25,1.75) -- (\x-1.25,1) -- (\x-1.5,1);
					}
					\node at (0,0.5) {$X_{AB}$};
					\node at (3,0.5) {$Y_{AB}$};
					\node at (6,0.5) {$Z_{AB}$};
					\foreach \x in {-1.75,7.75}
					{
						\draw (\x,0) circle (0.25);
						\draw (\x,1) circle (0.25);
					}
				\end{scope}
			\end{tikzpicture}
			&
			\begin{tikzpicture}[baseline,scale=0.5]
				\foreach \y in {0,1,2,3,4,5}
				{
					\draw[thick] (3,\y) -- (3.25,\y);
					\draw[thick] (4.75,\y) -- (5,\y);
					\draw[thick] (-1,\y) -- (1,\y);
					\draw[thick] (-3,\y) -- (-3.5,\y);
					\draw (-3.75,\y) circle (0.25);
					\draw[thick] (3,\y) -- (5.5,\y);
					\draw (5.75,\y) circle (0.25);
				}
				\draw[thick, rounded corners] (-3,4) -- (-2.5,4) -- (-1.5,2) -- (-1,2);
				\draw[thick,rounded corners] (-3,2) -- (-2.5,2) -- (-1.5,0) -- (-1,0);
				\draw[thick, rounded corners] (-3,0) -- (-2.5,0) -- (-1.5,4) -- (-1,4);
				\draw[thick] (-3,1) -- (-1,1);
				\draw[thick] (-3,3) -- (-1,3);
				\draw[thick] (-3,5) -- (-1,5);
				\draw[Blue,dashed,rounded corners] (-3,-0.25) rectangle (-1,5.25);
				\fill[Blue,opacity=0.2,rounded corners] (-3,-0.25) rectangle (-1,5.25);
				\node at (-2,4.5) {\textcolor{Blue}{$\overrightarrow{\Pi}_B$}};
				\draw[thick, rounded corners] (1,1) -- (1.5,1) -- (2.5,3) -- (3,3);
				\draw[thick, rounded corners] (1,3) -- (1.5,3) -- (2.5,5) -- (3,5);
				\draw[thick, rounded corners] (1,5) -- (1.5,5) -- (2.5,1) -- (3,1);
				\draw[thick] (1,0) -- (3,0);
				\draw[thick] (1,2) -- (3,2);
				\draw[thick](1,4) -- (3,4);
				\draw[Red,dashed,rounded corners] (1,-0.25) rectangle (3,5.25);
				\fill[Red, opacity=0.2,rounded corners] (1,-0.25) rectangle (3,5.25);
				\node at (2,0.5) {\textcolor{Red}{$\overleftarrow{\Pi}_{A}$}};
				\foreach \y in {0,2,4}
				{
					\draw[rounded corners,fill=white] (4-0.75,\y-0.25) rectangle (4+0.75,\y+1.25);
				}
				\node at (4,0.5) {$Z_{AB}$};
				\node at (4,2.5) {$Y_{AB}$};
				\node at (4,4.5) {$X_{AB}$};
			\end{tikzpicture}
			&
			\begin{tikzpicture}[baseline,scale=0.5]
				\draw[transparent] (0,-0.5) -- (0,5.5);
				\node at (0,2.5) {$\mathrm{Tr} \left( \textcolor{Blue}{\overrightarrow{\Pi}_B}\textcolor{Red}{\overleftarrow{\Pi}_{A}} X_{AB} \otimes Y_{AB} \otimes Z_{AB}\right)$};
				\node at (0,0.5) {\scriptsize $\textcolor{Blue}{\overrightarrow{\Pi}_B},\textcolor{Red}{\overleftarrow{\Pi}_{A}}$: cycle permutations};
			\end{tikzpicture}
			\\
			\hline
			\begin{tikzpicture}[baseline,scale=0.5]
				\draw[transparent] (0,-0.5) -- (0,5.5);
				\node at (0,2.5) {$\mathrm{Tr} \left( Y^{T_A}_{AB} X^{T_A}_{AB} Z^{T_A}_{AB} \right)$};
			\end{tikzpicture}
			&
			\begin{tikzpicture}[baseline,scale=0.5]
				\draw[transparent] (0,-0.5) -- (0,5.5);
				\begin{scope}[yshift=2cm]
					\draw[thick] (-1.5,0) -- (7.5,0);
					\foreach \x in {0,3,6}
					{
						\draw[rounded corners,fill=white] (\x-0.75,-0.25) rectangle (\x+0.75,1.25);
						\draw[thick, rounded corners] (\x-0.75,1) -- (\x-1,1) -- (\x-1,1.5) -- (\x+1.25,1.5) -- (\x+1.25,1) -- (\x+1.5,1);
						\draw[thick, rounded corners] (\x+0.75,1) -- (\x+1,1) -- (\x+1,1.75) -- (\x-1.25,1.75) -- (\x-1.25,1) -- (\x-1.5,1);
					}
					\node at (0,0.5) {$Y_{AB}$};
					\node at (3,0.5) {$X_{AB}$};
					\node at (6,0.5) {$Z_{AB}$};
					\foreach \x in {-1.75,7.75}
					{
						\draw (\x,0) circle (0.25);
						\draw (\x,1) circle (0.25);
					}
				\end{scope}
			\end{tikzpicture}
			&
			\begin{tikzpicture}[baseline,scale=0.5]
				\foreach \y in {0,1,2,3,4,5}
				{
					\draw[thick] (3,\y) -- (3.25,\y);
					\draw[thick] (4.75,\y) -- (5,\y);
					\draw[thick] (-1,\y) -- (1,\y);
					\draw[thick] (-3,\y) -- (-3.5,\y);
					\draw (-3.75,\y) circle (0.25);
					\draw[thick] (3,\y) -- (5.5,\y);
					\draw (5.75,\y) circle (0.25);
				}
				\draw[thick, rounded corners] (-3,5) -- (-2.5,5) -- (-1.5,3) -- (-1,3);
				\draw[thick, rounded corners] (-3,3) -- (-2.5,3) -- (-1.5,1) -- (-1,1);
				\draw[thick, rounded corners] (-3,1) -- (-2.5,1) -- (-1.5,5) -- (-1,5);
				\draw[thick] (-3,0) -- (-1,0);
				\draw[thick] (-3,2) -- (-1,2);
				\draw[thick](-3,4) -- (-1,4);
				\draw[Red,dashed,rounded corners] (-3,-0.25) rectangle (-1,5.25);
				\fill[Red, opacity=0.2,rounded corners] (-3,-0.25) rectangle (-1,5.25);
				\node at (-2,0.5) {\textcolor{Red}{$\overrightarrow{\Pi}_{A}$}};
				\draw[thick, rounded corners] (1,0) -- (1.5,0) -- (2.5,2) -- (3,2);
				\draw[thick, rounded corners] (1,2) -- (1.5,2) -- (2.5,4) -- (3,4);
				\draw[thick, rounded corners] (1,4) -- (1.5,4) -- (2.5,0) -- (3,0);
				\draw[thick] (1,1) -- (3,1);
				\draw[thick] (1,3) -- (3,3);
				\draw[thick] (1,5) -- (3,5);
				\draw[Blue,dashed,rounded corners] (1,-0.25) rectangle (3,5.25);
				\fill[Blue,opacity=0.2,rounded corners] (1,-0.25) rectangle (3,5.25);
				\node at (2,4.5) {\textcolor{Blue}{$\overleftarrow{\Pi}_B$}};
				\foreach \y in {0,2,4}
				{
					\draw[rounded corners,fill=white] (4-0.75,\y-0.25) rectangle (4+0.75,\y+1.25);
				}
				\node at (4,0.5) {$Z_{AB}$};
				\node at (4,2.5) {$Y_{AB}$};
				\node at (4,4.5) {$X_{AB}$};
			\end{tikzpicture}
			&
			\begin{tikzpicture}[baseline,scale=0.5]
				\draw[transparent] (0,-0.5) -- (0,5.5);
				\node at (0,2.5) {$\mathrm{Tr} \left( \textcolor{Red}{\overleftarrow{\Pi}_{A}}\textcolor{Blue}{\overrightarrow{\Pi}_B} X_{AB} \otimes Y_{AB} \otimes Z_{AB}\right)$};
				\node at (0,0.5) {\scriptsize $\textcolor{Red}{\overleftarrow{\Pi}_{A}}\textcolor{Blue}{\overrightarrow{\Pi}_B}
					= \left( \textcolor{Blue}{\overrightarrow{\Pi}_B}\textcolor{Red}{\overleftarrow{\Pi}_{A}}\right)^\dagger$};
			\end{tikzpicture} \\
			\hline
		\end{tabular}
	}
	\caption{\emph{Reformulations of relevant tensor product expressions:} The variance bounds in Sub.~\ref{sub:quadratic} and Sub.~\ref{sub:cubic} are contingent on bringing certain expressions into canonical form, i.e.\ $\mathrm{Tr} \left( O X_{AB} \otimes Y_{AB} \right)$ for bilinear functions and $\mathrm{Tr} \left( O' X_{AB} \otimes Y_{AB} \otimes Z_{AB} \right)$ for trilinear ones. This table supports visual derivations for these reformulations. Expressions of interest (very left) are first translated into wiring diagrams (center left). Subsequently, the rules of wiring calculus are used to re-arrange the diagrams (center right). 
		Translating them into formulas (very right) produces equivalent expressions that respect the desired structure.
	}
	\label{tab:diagrams}
\end{table*}

The arguments from the previous subsections 
make use of identities satisfied by traces of partial transposes of bipartite operators.
Wiring diagrams -- also known as tensor network diagrams -- provide a useful pictorial calculus 
for deriving such identifies. 
We refer the interested reader to Refs.~\cite{Landsberg2012,Bridgeman2017,Kueng2019} for a thorough introduction and content ourselves here with a concise overview that will suffice for the purposes at hand. 
The wiring formalism represents operators as boxes with lines emanating from them. These lines represent contra- (on the left) and co-variant indices (on the right):
\begin{align}
	X = \sum_{i,j} \left[X_{ij} \right]|i \rangle \! \langle j|
	=
	\begin{tikzpicture}[baseline,scale=0.5]
		\draw[thick] (-1.25,0) -- (1.25,0);
		\draw[rounded corners,fill=white] (-0.5,-0.5) rectangle (0.5,0.5);
		\node at (0,0) {$X$};
		\node at (-1,0.3) {\scriptsize \textcolor{gray}{$i$}};
		\node at (1,0.3) {\scriptsize \textcolor{gray}{$j$}};
	\end{tikzpicture}.
\end{align}
Two operators $X$ and $Y$ can be multiplied to produce another operator. 
This corresponds to an index contraction and is represented in the following fashion:
\begin{align}
	XY =& \sum_{i,k} ( \sum_j \left[X\right]_{ij} \left[ X \right]_{jk} ) |i \rangle \! \langle k|
	=
	\begin{tikzpicture}[baseline,scale=0.5]
		\draw[thick] (-2,0) -- (2,0);
		\draw[rounded corners,fill=white] (-1.5,-0.5) rectangle (-0.5,0.5);
		\node at (-1,0) {$X$};
		\draw[rounded corners, fill=white] (0.5,-0.5) rectangle (1.5,0.5);
		\node at (1,0) {$Y$};
		\node at (0,0.3) {\scriptsize \textcolor{gray}{$j$}};
		\node at (-2,0.3) {\scriptsize \textcolor{gray}{$i$}};
		\node at (2,0.3) {\scriptsize \textcolor{gray}{$k$}};
	\end{tikzpicture}.
\end{align}
Transposition exchanges outgoing (contravariant) and incoming (covariant) indices
\begin{align}
	X^T
	=& \sum_{i,j} \left[X\right]_{ij} |j \rangle \! \langle i|
	=
	\begin{tikzpicture}[baseline,scale=0.5]
		\draw[rounded corners,fill=white] (-0.5,-0.5) rectangle (0.5,0.5);
		\node at (0,0) {$X$};
		\draw[thick,rounded corners] (-0.5,0) -- (-0.75,0) -- (-0.75,0.75) -- (1,0.75) -- (1,0) -- (1.5,0);
		\draw[thick, rounded corners] (0.5,0) -- (0.75,0) -- (0.75,1) -- (-1,1) -- (-1,0) -- (-1.5,0);
		\node at (-1.5,0.3) { \scriptsize \textcolor{gray}{$j$}};
		\node at (1.5,0.3) {\scriptsize \textcolor{gray}{$i$}};
	\end{tikzpicture} ,
\end{align}
while the trace pairs up both indices and sums over them:
\begin{align}
	\mathrm{Tr}(X) =& \sum_{i} \left[X\right]_{ii}
	= \begin{tikzpicture}[baseline,scale=0.5]
		\draw[rounded corners,fill=white] (-0.5,-0.5) rectangle (0.5,0.5);
		\node at (0,0) {$X$};
		\node at (0,1.05) {\scriptsize \textcolor{gray}{$i$}};
		\draw[thick,rounded corners] (-0.5,0) -- (-1.25,0) -- (-1.25,0.75) -- (1.25,0.75) -- (1.25,0) -- (0.5,0);
	\end{tikzpicture}
	=
	\begin{tikzpicture}[baseline,scale=0.5]
		\draw[thick,rounded corners] (-1.25,0) -- (1.25,0);
		\draw[rounded corners,fill=white] (-0.5,-0.5) rectangle (0.5,0.5);
		\node at (0,0) {$X$};
		\draw (-1.5,0) circle (0.25);
		\draw (1.5,0) circle (0.25);
	\end{tikzpicture}.
\end{align}
We abbreviate this loop (contraction of leftmost and rightmost indices) by putting two circles at the end points of lines that should be contracted. This notation is not standard, but  will considerably increase the readability of more complex contraction networks.

This basic formalism readily extends to tensor products if we arrange tensor product factors in parallel. For instance, a bipartite operator features two parallel lines on the left and on the right:
\begin{align}
	X_{AB} 
	=&
	\begin{tikzpicture}[baseline,scale=0.5,yshift=-0.25cm]
		\foreach \y in {0,1}
		{\draw[thick] (-1.5,\y) -- (-0.75,\y);
			\draw[thick] (0.75,\y) -- (1.5,\y);
		}
		\draw[rounded corners, fill=white] (-0.75,-0.25) rectangle (0.75,1.25);
		\node at (0,0.5) {$X_{AB}$};
		\node at (-2,1) {\scriptsize \textcolor{gray}{$A$}};
		\node at (2,1) {\scriptsize \textcolor{gray}{$A$}};
		\node at (-2,0) {\scriptsize \textcolor{gray}{$B$}};
		\node at (2,0) {\scriptsize \textcolor{gray}{$B$}};
	\end{tikzpicture}
\end{align}
The upper lines represent the system $A$, while the lower lines represent system $B$. 
Two important bipartite operators are the identity $\mathbb{I}$ (do nothing) and the swap operator $\Pi$ that exchanges the systems:
\begin{align}
	\begin{tikzpicture}[baseline,scale=0.5,yshift=-0.25cm]
		\foreach \y in {0,1}
		{\draw[thick] (-1.5,\y) -- (1.5,\y);
		}
		\draw[rounded corners, fill=white,opacity=0.5] (-0.75,-0.25) rectangle (0.75,1.25);
		\node at (0,0.5) {$\mathbb{I}$};
	\end{tikzpicture}
	=
	\begin{tikzpicture}[baseline,scale=0.5,yshift=-0.25cm]
		\foreach \y in {0,1}
		{\draw[thick] (-1,\y) -- (1,\y);
		}
	\end{tikzpicture}
	\quad \text{and} \quad 
	\begin{tikzpicture}[baseline,scale=0.5,yshift=-0.25cm]
		\foreach \y in {0,1}
		{\draw[thick] (-1.5,\y) -- (-0.75,\y);
			\draw[thick] (0.75,\y) -- (1.5,\y);
		}
		\draw[rounded corners, fill=white] (-0.75,-0.25) rectangle (0.75,1.25);
		\node at (0,0.5) {$\Pi$};
	\end{tikzpicture}
	=
	\begin{tikzpicture}[baseline,scale=0.5,yshift=-0.25cm]
		\draw[thick, rounded corners] (-1,1) -- (-0.5,1) -- (0.5,0) -- (1,0);
		\draw[thick, rounded corners] (-1,0) -- (-0.5,0) -- (0.5,1) -- (1,1);
	\end{tikzpicture}.
\end{align}
Rules for multiplying and contracting operators readily extend to the tensor setting. This allows us to reformulate well-known expressions pictorially. For instance,
\begin{align}
	\mathrm{Tr}(XY)
	= &
	\begin{tikzpicture}[baseline,scale=0.5,yshift=0.25cm]
		\draw[thick] (-2,0) -- (2,0);
		\draw[rounded corners,fill=white] (-1.5,-0.5) rectangle (-0.5,0.5);
		\node at (-1,0) {$X$};
		\draw[rounded corners, fill=white] (0.5,-0.5) rectangle (1.5,0.5);
		\node at (1,0) {$Y$};
		\draw (2.25,0) circle (0.25);
		\draw (-2.25,0) circle (0.25);
	\end{tikzpicture}
	=
	\begin{tikzpicture}[baseline,scale=0.5,yshift=-0.25cm]
		\draw[rounded corners,fill=white]
		(-0.5,-0.75) rectangle (0.5,0.25);
		\draw[rounded corners, fill=white]
		(-0.5,0.75) rectangle (0.5,1.75);
		\node at (0,1.25) {$X$};
		\node at (0,-0.25) {$Y$};
		\draw[thick, rounded corners]
		(-1,1.25) -- (-1.5,1.25) -- (-2.5,-0.25) -- (-3,-0.25);
		\draw[thick, rounded corners] (-1,-0.25) -- (-1.5,-0.25) -- (-2.5,1.25) -- (-3,1.25);
		\foreach \y in {-0.25,1.25}
		{
			\draw[thick] (-3.5,\y) -- (-3,\y);
			\draw[thick] (-1,\y) -- (-0.5,\y);
			\draw[thick] (0.5,\y) -- (1,\y);
			\draw (1.25,\y) circle (0.25);
			\draw (-3.75,\y) circle (0.25);
		} 
	\end{tikzpicture} \\
	=& \mathrm{Tr} \left( \Pi X \otimes Y \right).
\end{align}
The wiring formalism is also exceptionally well-suited to capture partial operations, like the partial transpose:
\begin{align}
	X_{AB}^{T_A} =&
	\begin{tikzpicture}[baseline,scale=0.5,yshift=-0.5cm]
		\draw[thick, rounded corners] (-0.75,1) -- (-1,1) -- (-1,1.5) -- (1.25,1.5) -- (1.25,1) -- (1.5,1);
		\draw[thick, rounded corners] (0.75,1) -- (1,1) -- (1,1.75) -- (-1.25,1.75) -- (-1.25,1) --(-1.5,1);
		\draw[thick] (-1.5,0) -- (1.5,0);
		\draw[rounded corners,fill=white] (0-0.75,-0.25) rectangle (0+0.75,1.25);
		\node at (0,0.5) {$X_{AB}$};
	\end{tikzpicture}.
\end{align}
These elementary rules can be used to visually represent more complicated expressions -- like a trace of multiple partial transposes. The wiring formalism provides a pictorial representation for such objects and a visual framework for modifying them. In particular, it is possible to bend, as well as unentangle, index lines and rearrange tensor factors at will.
Table~\ref{tab:diagrams} collects several such modifications that are important for the arguments above.

\end{document}